\documentclass[a4paper]{article}

\usepackage{amsmath,amsthm,amssymb,mathtools}
\usepackage{comment}
\usepackage[T1]{fontenc} 
\usepackage[mathcal]{euscript}
\usepackage{graphicx}
\usepackage[hidelinks]{hyperref}
\usepackage[top=2.8cm,bottom=2.8cm,left=2.8cm,right=2.8cm]{geometry}
\usepackage{tikz}
\usepackage[export]{adjustbox}
\usepackage{appendix}
\usepackage{setspace}
\usepackage{tcolorbox}
\usepackage[normalem]{ulem}
\numberwithin{equation}{section}
\theoremstyle{plain}
\newtheorem{theorem}{Theorem}[section]
\newtheorem{corollary}{Corollary}[theorem]
\theoremstyle{definition}

\makeatletter
\newcommand*{\rom}[1]{\expandafter\@slowromancap\romannumeral #1@}
\makeatother
\makeatletter
\g@addto@macro\bfseries{\boldmath}
\makeatother

\usetikzlibrary{shapes.geometric,decorations.markings,intersections,calc,positioning}

\tikzset{->-/.style={decoration={
  markings,
  mark=at position .6 with {\arrow{>}}},postaction={decorate}}}

\tikzset{-<-/.style={decoration={
  markings,
  mark=at position .6 with {\arrow{<}}},postaction={decorate}}}

\tikzset{%
    add/.style args={#1 and #2}{
        to path={%
 ($(\tikztostart)!-#1!(\tikztotarget)$)--($(\tikztotarget)!-#2!(\tikztostart)$)%
  \tikztonodes},add/.default={.2 and .2}}
}  

\tikzset{
    extended line/.style={shorten >=-#1,shorten <=-#1},
    extended line/.default=1cm]
}

\tikzset{line through/.style args={#1 parallel to line through #2 and #3 and
length #4}{insert path={%
let \p1=($(#3)-(#2)$),\n1={atan2(\y1,\x1)} in (#1) -- ++ (\n1:#4)}}}

\definecolor{col1}{rgb}{0.4, 0.69, 0.2}
\definecolor{col2}{rgb}{0.96, 0.29, 0.54}

\definecolor{green(ryb)}{rgb}{0.4, 0.69, 0.2}
\definecolor{frenchrose}{rgb}{0.96, 0.29, 0.54}
\definecolor{persianblue}{rgb}{0.11, 0.22, 0.73}
\definecolor{jade}{rgb}{0.0, 0.66, 0.42}
\definecolor{limegreen}{rgb}{0.2, 0.8, 0.2}

\title{Quadratic Bureau-Guillot  
systems  with the first and second Painlev\'e transcendents  in the coefficients. \\ Part I:   geometric approach and  birational equivalence}
\author{Marta Dell'Atti{\color{jade}$\,^1$} \\[-.5ex] \small{\url{m.dell-atti@uw.edu.pl}} \\[1ex] Galina Filipuk{\color{jade}$\,^1$} \\[-.5ex]  \small{\url{g.filipuk@uw.edu.pl}} \\[1ex]
{\color{jade}$\,^1$}{ \small University of Warsaw, Institute of Mathematics, 
}  \\[.1ex] {\small Banacha 2, Warsaw, 02-097, Poland}
}

\date{}

\begin{document}

\maketitle

\begin{abstract}
Bureau proposed a classification of systems of quadratic differential equations in two variables which are free of movable critical points, which was recently revised by Guillot. We revisit the quadratic Bureau-Guillot systems with the first and second Painlev\'e transcendent in the coefficients. We explain  their birational equivalence by using the geometric approach of Okamoto's spaces of initial conditions and the method of iterative polynomial regularisation, solving the Painlev\'e equivalence problem for the Bureau-Guillot systems with non-rational meromorphic coefficients. {We explicitly determine the Hamiltonian functions associated with the systems, both for the cases where the system is Hamiltonian with respect to the standard $2$-form and those for which are Hamiltonian with respect to the pull back induced by certain changes of variables.} We also find that one of the systems related to the second Painlev\'e equation can be transformed into a Hamiltonian system via the iterative polynomial regularisation. 
\end{abstract}









    



\section{Introduction}

It is well-known that certain linear second order ordinary differential  equations, like the Lam\'e or Heun equation \cite{Heun}, are endowed with the so-called elliptic form, i.e.\ they can be written as linear equations with elliptic coefficients. At the end of the nineteenth century, Halphen studied a sequence of higher-order linear equations with doubly periodic coefficients having the property that quotients of solutions are single valued \cite{Halphen}. This result was then generalised in~\cite{Pickering}, where the authors considered equations with the first Painlev\' e transcendent ($\text{P}_{\text{I}}$) in the coefficients. In \cite{Pickering}  higher order systems of nonlinear equations with the Painlev\'e property in the coefficients were also obtained. For a non-linear system to possess the Painlev\'e property means that its solutions are free from movable critical points, admitting only poles as movable singularities. 
The Painlev\'e equations as well can be represented in an elliptic form, starting from the sixth Painlev\'e equation ($\text{P}_{\text{VI}}$) and proceeding via successive coalescences to obtain the form for the remaining $\text{P}_{\text{V}}$, $\dots$, $\text{P}_{\text{I}}$  (see~\cite{Manin, P6 ell,Takasaki} and the references therein). The relevance of studying the differential equations with special functions in the coefficients -- although such forms could be more complicated than equations with rational coefficients -- lies on the possibility to explore new connections with other areas of mathematics (e.g.\ algebraic geometry, number theory, and so on). 


In a series of papers (see \cite{Bureau table} and references in \cite{Guillot}), following the classical works of Painlev\'e and Gambier, Bureau classified quadratic differential equations in two variables endowed with the Painlev\'e property. This classification, summarized in the table in \cite{Bureau table}, was recently revised and extended by Guillot in~\cite{Guillot}, where the author also investigated the birational equivalence among some of the systems in the list \cite[Proposition 2]{Guillot}. Bureau as well already mentioned some of the birational transformations among the systems, using elaborate changes of variables between pairs of differential equations. 

In this work we are interested to explain the above mentioned birational transformations in a unified way using the geometric theory of the Painlev\'e equations. With this aim, we study some of the Bureau-Guillot non-autonomous systems of the form 
\vspace*{-1ex}
\begin{equation}
    \begin{cases}
        y'=F\!\big(y(x),z(x);x\big)  \\[1ex]
        z'=G\!\big(y(x),z(x);x\big)  
    \end{cases}
\end{equation}
where $x \in \mathbb{C}$, the field variables $(y(x),z(x))$ have values in $\mathbb{C}^2$, and $F$, $G$ are polynomial functions in the field variables $(y,z)$ with meromorphic coefficient functions in $x$, not necessarily rational. In particular, we focus on those systems which are related in some way to the first Painlev\'e equation ($\text{P}_{\text{I}}$) and the second Painlev\'e equation ($\text{P}_{\text{II}}$), that we recall here: 
\begin{align}
    \text{P}_{\text{I}} &\colon y''= 6y^2+ x \label{eq:P1} \\[1ex]
    \text{P}_{\text{II}} &\colon y''= 2y^3 + yx + \alpha \label{eq:P2}
\end{align}
where $y=y(x)$ and $\alpha \in \mathbb{C}$ is a constant. 
We use the notation fixed by Guillot in the labelling of the quadratic systems in~\cite{Guillot}, and we will refer to these as Bureau-Guillot systems. In particular, we will focus on the systems V, IX.B(2), IX.B(3), IX.B(5), XIII and XIV. Sometimes we shall use a suitable labelling of the variables to refer to a particular system to avoid confusion, especially when directly compared, e.g.~$(z_5,y_5)$ for the field variables satisfying the system~V and~$(z_{92},y_{92})$ for the system~IX.B(2).  

We will establish the birational equivalences of the systems in the list by using two different approaches, namely the geometric approach constructing the Okamoto  space of initial conditions associated with the systems \cite{Okamoto} and generalised by Sakai \cite{Sakai}, and the iterative polynomial regularisation introduced in \cite{GF}. For details on the geometric approach see \cite{Okamoto, KNY, DFS, DFLS, MDTK2}.
There are several interesting questions we would like to address in this and subsequent papers. In this work (Part I), we investigate why it is possible to establish a birational equivalence between certain pairs of the selected Bureau-Guillot systems mentioned above. Another interesting question related to the selected Bureau-Guillot systems pertains the value distribution theory, and in particular the study of the behaviour of solutions, especially those with non-rational meromorphic coefficient functions. This is the focus of the Part II \cite{ECGF}.
Here, we tackle the problem of the existence of birational transformations by examining the geometry behind the systems put in relation, and show that they share the same surface type (in the sense of the Okamoto-Sakai theory of Painlev\'e equations) by exploiting the geometric approach described below. This allows us to explicitly reproduce the above mentioned birational transformations, which in some cases can additionally be obtained via the method of the iterative regularisation. 

Not all the systems studied in this work are Hamiltonian. For those which are Hamiltonian, we will explicitly write its polynomial expression also specifying its genus. This is a slight abuse of terminology, since whenever we write ``Hamiltonian of genus $g$'' we refer to the genus associated with the corresponding Newton polygon, following \cite{MDTK2}. A polynomial Hamiltonian defines an algebraic curve in~$\mathbb{P}^2$ with which we can associate a Newton polygon, given by the convex hull of the set of points $(j,k)$ with $\beta_{jk} \neq 0$, e.g.
\vspace*{-3ex}
\begin{equation}\label{eq:ham_example}
   \begin{aligned}
H(z,y)&= \sum_{jk}\beta_{jk}(x)\,z^jy^k\,, \\[1ex]
H(z,y)&=z^3y+\beta_{30}\,y^3+\beta_{20}\,z^2+\beta_{11}\,zy+\beta_{02}\,y^2+\beta_{10}\,z+\beta_{01}\,y
\end{aligned} \qquad \includegraphics[width=.15\textwidth,valign=c]{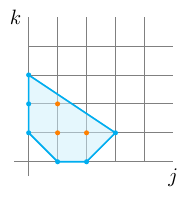} 
\end{equation}
The genus of the curve corresponds to the number of internal points of the Newton polygon, hence e.g.\ the Hamiltonian in~\eqref{eq:ham_example} has genus 3. 
As shown in~\cite{new}, the genus and the area of the Newton polygon might change during the (iterative) polynomial regularisation, and the study of different shapes of Newton polygons can be helpful when approaching the Painlev\'e equivalence problem.

The paper is organised as follows. In Section~\ref{sec:geom_iter} we describe the main approaches used to analyse the systems, i.e.\ the geometric approach and the iterative regularisation method. In Section~\ref{sec:P1} we study in detail the geometry of the Bureau-Guillot systems underpinning $\text{P}_{\text{I}}$, namely the systems~V, IX.B(2), IX.B(5) and XIII as listed in~\cite{Guillot}. We associate with the systems a specific surface type and we find the suitable birational transformations between pairs of systems. For the system IX.B(2) we determine the Hamiltonian form with respect to the non-standard $2$-form induced by the relation with the variables for the system V. In Section~\ref{sec:P2} we analyse the Bureau-Guillot systems underpinning $\text{P}_{\text{II}}$, i.e.\ the system IX.B(3) and XIII and establish the birational transformation relating them. {We provide a description of the system XIII being Hamiltonian with respect to the pull back of the $2$-form expressed in terms of the variables solving the system IX.B(3).} For the system XIII we also find via the iterative regularisation a birationally related Hamiltonian system. The Section~\ref{sec:disc} is devoted to some discussion and to explore several interesting open questions.

\section{Geometric approach and iterative regularisation method}\label{sec:geom_iter}
We follow the approach established in \cite{DFLS}, based on the Okamoto-Sakai construction, to connect several Hamiltonian systems associated with the fourth Painlev\'e equation  $\text{P}_{\text{IV}}$ and find the suitable changes of coordinates and parameters relating pairs of systems. The method has been successfully applied to several problems (see e.g.~\cite{DFS,MDTK1}), and we use it here to establish birational transformations between pairs of related systems.    
In our case, we deal with systems (either Hamiltonian or not Hamiltonian) underpinning either $\text{P}_{\text{I}}$ or $\text{P}_{\text{II}}$. 

We identify the space of initial conditions as it was introduced by Okamoto \cite{Okamoto} by considering first an extension of the affine space for the variables $(y,z)\in \mathbb{C}^2$ to include points at infinity (i.e.~$\mathbb{P}^2$) and then removing from it the configuration of vertical leaves obtained in the regularisation process. The minimal configuration of vertical leaves takes the name of \emph{surface type} and it is represented by an extended Dynkin diagram (see e.g.~\cite{Matano}). We briefly review this procedure to fix the notation and the compactification we use throughout the paper.


Among the possible compactifications usually considered for this problem (i.e.\ $\mathbb{P}^2$, $\mathbb{P}^1 \times \mathbb{P}^1$, or a Hirzebruch surface) we choose to consider $\mathbb{P}^2$, to deal with the simplest form of the expressions for the change of coordinates (one type of curve to start with instead of at least two). 
The compactified space $\mathbb{P}^2$ is given by gluing together three affine charts identified by $(z,y)$, $(u_0,v_0)$ and $(U_0,V_0)$, related by homogeneous coordinates:
\vspace*{-3ex}
\begin{equation}\label{eq:CP2}
\includegraphics[width=.23\textwidth,valign=c]{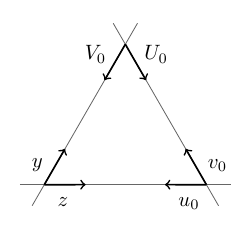} \qquad   
    \begin{aligned}
    &\mathbb{P}^2 = \mathbb{A}_{(z,y)} \cup  \mathbb{A}_{(u_0,v_0)} \cup \mathbb{A}_{(U_0,V_0)}\,, \\[1.4ex] 
    &[\,1:z:y\,] = [\,u_0:1:v_0\,] = [\,V_0:U_0:1\,]\,, \\[.7ex]
    &u_0 = \frac{1}{z}\,, \qquad V_0 = \frac{1}{y}\,, \qquad v_0 = \frac{y}{z} = \frac{1}{U_0} \,.
\end{aligned}
\end{equation}

\vspace*{-2ex}

\noindent 
In $\mathbb{P}^2$, the system is viewed by all respective charts ($(z,y)$, $(u_0,v_0)$ and $(U_0,V_0)$), may have points of indeterminacy, i.e.\ points at which the vector field has the ill-defined form $0/0$. Some of these points are \emph{base points}, i.e.\ points at which infinitely many flow curves of the vector field coalesce. The system of equations is regularised through the process of blowing-up the phase space at these base points, giving rise to an exceptional curve for each blow-up. In a coordinate chart $(u_i,v_i)$, a blow-up at a point $p\colon (u_i,v_i)=(a,b)$ is performed by 
\begin{equation}
  \text{Bl}_p(\mathbb{C}^2) = \{ (u_i,v_i) \times [w_0:w_1] \in \mathbb{C}^2 \times \mathbb{P}^1: (u_i-a) w_0 = (v_i-b) w_1 \}\,,
\end{equation}
giving rise to two new coordinate charts, namely
$(u_j,v_j)$, $(U_j,V_j)$, according to 
\begin{equation}\label{eq:blowup_coords}
\begin{cases}
    u_i = u_j+a = U_j V_j+a,  \\[1ex] 
    v_i = u_j\, v_j+b = V_j+b.
\end{cases}.
\end{equation}
Here, $U_j = {w_0}/{w_1}$ covers the part of $\mathbb{P}^1$ modulo where $w_1 \neq 0$, while $v_j = {w_1}/{w_0}$ covers the part where $w_0 \neq 0$. The projection onto the first component is then
\begin{equation}
    \pi_p : \text{Bl}_p(\mathbb{C}^2) \to \mathbb{C}^2, \quad (u_i,v_i) \times [w_0,w_1] \mapsto (u_i,v_i).
\end{equation}
Away from the point $p$, the blow-up is a one-to-one map,
\begin{equation}
    \text{Bl}_p(\mathbb{C}^2) \setminus \pi_p^{-1}(p) \quad \longleftrightarrow \quad \mathbb{C}^2 \setminus \{ p\}.
\end{equation}
The set $E = \pi_p^{-1}(p)$ is the exceptional curve introduced by the blow-up which as a point set in coordinates is given by 
\begin{equation}
    E = \{ u_j = 0 \} \cup \{ V_j = 0 \}.
\end{equation}
After a blow-up, new base points can arise on the exceptional curve and usually, to completely remove the indeterminacy, the composition of finitely many successive changes of coordinates is required.

Following \cite{Sakai} the divisors on the phase space give rise to the Picard group of the $n$-times blown up space~$\mathcal{X}_n$, 
\begin{equation}
    \text{Pic}(\mathcal{X}_n) = \text{Span}_{\mathbb{Z}} \{ \mathcal{H} , \mathcal{E}_1, \dots,  \mathcal{E}_n\}\,,
\end{equation}
where $\mathcal{H}$ in~\eqref{eq:CP2} is the hyperplane coordinates divisor in $\mathbb{P}^2$, and $\mathcal{E}_i$ are the exceptional curves introduced by the $n$ blow-ups. By definition, exceptional curves have self-intersection $-1$, while $\mathcal{H}$ itself has self-intersection $+1$. The anti-canonical divisor class is then 
\begin{equation}\label{eq:anticanon}
    -\mathcal{K}_{\mathcal{X}_n} = 3\,\mathcal{H} - \mathcal{E}_1 - \dots - \mathcal{E}_{n}\,,
\end{equation}
and surfaces of different
types are associated with different configurations of irreducible components of the anti-canonical divisor. 
The linear extension of the self-intersection relations gives rise to the intersection form
\begin{equation}\label{eq:intersection}
\mathcal{H} \cdot \mathcal{H} = 1, \qquad \mathcal{E}_i \cdot \mathcal{H} = 0 \quad\forall i, \qquad \mathcal{E}_i \cdot \mathcal{E}_j = -\delta_{ij}\,. \end{equation}
Crucial in the construction of the configuration is considering the effect that the blow-up transformation has on the curves involved. We distinguish two cases: 
\begin{enumerate}
    \item The point of indeterminacy lies on the intersection of two irreducible components of the inaccessible divisor $L_1$ and $L_2$
    
    \vspace*{-4ex}
    
    \begin{equation*}
        \includegraphics[width=.58\textwidth]{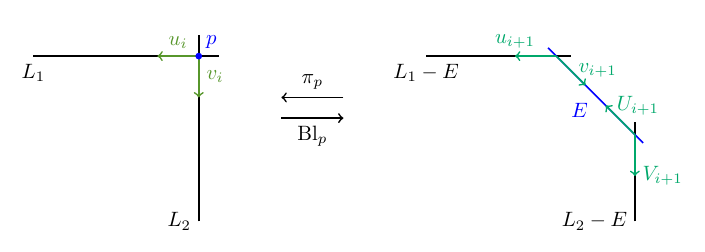}
    \end{equation*}
    
    \vspace*{-2ex}
    
    In this case, the base point is at the origin of the chart $(u_i,v_i)=(0\,,0)$ and not visible in other charts. After the blow-up, the irreducible divisor components of the total transform $\hat{L}_1 \cup \hat{L}_2$ are $\hat{L}_1 - E$, $E$ and $\hat{L}_2 - E$, with $E$ exceptional curve. The notation $L_i-E$ stands for the strict transform of $L_i$ under the blow-up, i.e.\ 
    \begin{equation*}
        L_i-E = \overline{\pi_p^{-1}\{L_i \setminus\{p\}\}} \,.
    \end{equation*}
    In the following we will simply write $L_i-E$ by abuse of notation. 

    \item The point of indeterminacy is not an intersection point of inaccessible divisor components 
     
    \vspace*{-2ex}
    
    \begin{equation*} 
        \includegraphics[width=.58\textwidth]{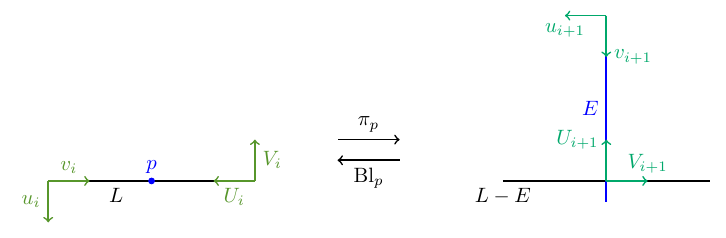}
    \end{equation*}
    
    \vspace*{-2ex}
    
\noindent    hence, it will be visible in both charts 
    $$(u_i,v_i) = (0\,,a)\,, \qquad (U_i,V_i) = (1/a,0)\,, \quad a \in \mathbb{C}\,.$$
\end{enumerate}

\vspace*{-2ex}

\noindent 
In both cases the effect of the blow-up transformation on the curve (or curves), is that their self-intersection value changes. For instance, a curve $H$ of the class $\mathcal{H}$ has self-intersection $+1$, but after a blow-up of a point on $H$, the self-intersection of the curve $H-E$ (with $E\in \mathcal{E}$) is determined by using~\eqref{eq:intersection}, i.e.\
\begin{equation}
    (H-E) \cdot (H-E) = H \cdot H - H\cdot E - E \cdot H + E \cdot E = 0\,.
\end{equation}
As already mentioned, usually the process of regularisation occurs in several steps, hence after a blow-up it is usual to find a new point,  blowing up  which produces  the exceptional curve, say $E_i$, on the rational surface. After the subsequent blow-up   an exceptional curve, say $E_j$, intersects $E_i$. The self intersection of the curve $E_i - E_j$ is again evaluated making use of the intersection product~\eqref{eq:intersection}
\begin{equation}\label{eq:example_2}
    (E_i - E_j)\cdot (E_i - E_j)= E_i\cdot E_i - E_i\cdot E_j - E_j\cdot E_i + E_j \cdot E_j = -2\,. 
\end{equation}
By tracking   the subsequent points at which the surface is blown up, we can visualise the configuration of the irreducible components of the anti-canonical divisor~\eqref{eq:anticanon} representing the surface type associated with the system. In particular, the Dynkin diagram $\delta=-\mathcal{K}_{\mathcal{X}_n}$ can be identified by considering the curves with self-intersection~$-2$. In the following, the diagrams associated with the different configurations of intersecting curves are drawn with different colours, indicating different self-intersections:
\vspace*{-1ex}
\begin{equation*}
\begin{tikzpicture}[scale=.9]
    \draw[thick,jade] (0,0) -- ++(1,0) node[right] {\small $-2$-curve};
    \draw[thick,blue!80!black] (4,0) -- ++(1,0) node[right] {\small $-1$-curve};
    \draw[thick,gray!80!black] (8,0) -- ++(1,0) node[right] {\small $\ge 0$-curve};
\end{tikzpicture}
\end{equation*}

\vspace*{-2ex}

\noindent 
and e.g.\ the curve in~\eqref{eq:example_2} would be represented by a green line.  
To each $-2$-curve corresponds a node, and two nodes are connected by an edge if the corresponding curves intersect. The $k$-th node in the diagram also represents $\delta_k$, the class of the irreducible component of the inaccessible divisor. The systems treated in this paper will be associated with the extended Dynkin diagrams~$E_7^{(1)}$ and~$E_8^{(1)}$, i.e.  
\begin{equation*}
    \includegraphics[width=.3\textwidth,valign=c]{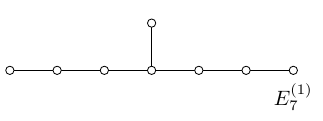} \hspace{5ex} \includegraphics[width=.32\textwidth,valign=c]{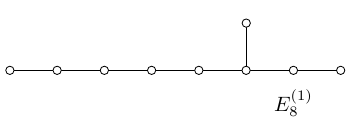}
\end{equation*}
respectively related to $\text{P}_{\text{II}}$ and $\text{P}_{\text{I}}$.

Once we have a minimal intersection diagram associated with a pair of systems -- say in coordinates of the original affine chart $(z,y)$ and $(w,t)$ respectively -- we can proceed to matching the irreducible components by identifying pairs of curves~\cite{DFLS}. This identification allows us to establish a preliminary form of the functional expression for $(z,y)$ in terms of $(w,t)$, and vice versa. This is accomplished by looking for the form of the axes (e.g.\ the sets $\{z=0\}$ and $\{y=0\}$ for the first system). In the following we will see cases where the functional relation between coordinates is either linear, quadratic or at most cubic, as expected at least for systems which are Hamiltonian associated with the Painlev\'e equations \cite{KMNOY}. At this stage, not all the coefficient functions at all orders are fixed, and to completely determine them we require that the variables  $(z(w,t),y(w,t))$ given as functions of the coordinates~$(w,t)$ still satisfy the system for $(z,y)$. 

The second method we use to look for birational transformations is the iterative regularisation, as introduced in \cite{GF}, which focuses more on the analytical aspects of the problem. The regularisation of the system is conducted as described above, via chains of blow-up transformations originating at each point of indeterminacy in~$\mathbb{P}^2$. If the system possesses the Painlev\'e property, it is regularised after a finite number of blow-ups for each cascade. Since there might be more than one cascade of blow-ups to be considered we obtain several versions of the regularised system, one per each tail of the chains of transformations of coordinates. The idea behind the method is to iteratively regularise the systems as they appear in the final affine charts. We promote each of the systems in the final charts to be the original system, introduce a suitable compactification and we consider again the subsequent points of indeterminacy to find new regularised versions of it. In this way the algorithmic procedure of the iterative regularisation produces a tree where the root node is the original system, and the leaves are the different versions of the regularised systems. Along some of these ramifications, we can produce a significantly simpler form of the system compared to that of the original one, or -- as we will see below -- we can recognise at some step a different system to which our original one is related via a birational transformation. This approach alongside with the analysis of the Newton polygons associated with polynomial Hamiltonian systems is deepened in \cite{new}. 

In the following, the geometric approach and the identification of the irreducible components of the anti-canonical divisors will be crucial to establish the birational transformations to relate pairs of the Bureau-Guillot systems among V, IX.B(2), IX.B(5) and XIV for $\text{P}_{\text{I}}$, and to relate the Bureau-Guillot systems IX.B(3) and XIII for $\text{P}_{\text{II}}$. The birational transformations relating the systems IX.B(2) and IX.B(5) and the systems IX.B(3) and XIII are also obtained via the iterative regularisation.

\section{\texorpdfstring{Bureau-Guillot systems related to $\text{P}_{\text{I}}$}{P1}}\label{sec:P1}

All the systems related to the first  Painlev\'e equation admit a surface type represented by the extended Dynkin diagram $E_8^{(1)}$ with the labelling for the components $\delta_i$ as 
\vspace*{-2ex}
\begin{equation}\label{eq:diag_E8}
    \includegraphics[width=.37\textwidth,valign=c]{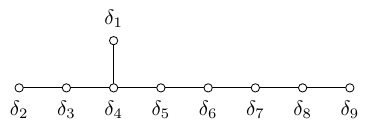}
\end{equation}
Each system is characterised by a specific configuration of $\delta_i$  
obtained during the regularisation process. In order to obtain the transformations of variables between systems  
we compare these configurations.    

\subsection{System V}
The system V, which will be our reference system for  finding birational transformations with meromorphic coefficients to other Bureau systems in case of the first Painlev\'e transcendent, reads as follows:
\begin{equation}\label{eq:syst_5}
\begin{cases}
    y' = z,\\[1ex]
z' = 6y^2 + f(x),
\end{cases} \qquad f''(x) = 0\,.
\end{equation}
We consider the case when $f(x)=x$ so that the function $y$ satisfies the standard first Painlev\'e equation~\eqref{eq:P1}.  
The system V in~\eqref{eq:syst_5} is Hamiltonian with the Hamiltonian function of genus 1 and area equal to 2
\vspace*{-3ex}
\begin{equation}
H_{5}(z,y)=\frac{z^2}{2}-y \left(x+2 y^2\right) \hspace{10ex} \includegraphics[width=.16\textwidth,valign=c]{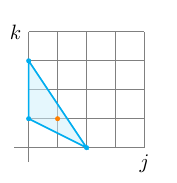}   
\end{equation}

\vspace*{-2ex}

\noindent
A  polynomial regularisation of the system~\eqref{eq:syst_5} is  described in \cite{Iwasaki P1}  and it requires an intermediate 2-fold covering transformation.  The rational regularisation can be realised without using the intermediate 2-fold covering \cite{Okamoto, DJ, KNY}.  

The rational regularisation of the system V produces a single chain of transformations of variables, the cascade with $9$ blow-ups from $\mathbb{P}^2$ for $\text{P}_{\text{I}}$, originating at $(u_0,v_0)$
\begin{equation}
\vcenter{\hbox{
	\begin{tikzpicture}
		\begin{scope} 
			\path (0,0) coordinate (a0) -- ++(3,0)  coordinate (a1) -- ++(3,0)  coordinate (a2) -- ++(3,0)  coordinate (a3) -- ++(0,-1)  coordinate (a4)  -- ++(0,-1)  coordinate (a5) -- ++(-3,0)  coordinate (a6) -- ++(-3.2,0)  coordinate (a7) -- ++(-3.4,0)  coordinate (a8);
			\node (b0) at (a0) {$(u_0,v_0)=(0\,, 0)$};
			\node (b1) at (a1) {$(U_1,V_1)=(0\,, 0)$};
			\node (b2) at (a2) {$(U_2,V_2)=(0\,, 0)$};
			\node (b3) at (a3) {$(U_3,V_3)=( 4\,, 0)$};
			\node (b4) at (a4) {$(U_4,V_4)=(0\,, 0)$};
			\node (b5) at (a5) {$(U_5,V_5)=(0\,, 0)$};
			\node (b6) at (a6) {$(U_6,V_6)=(0\,, 0)$};
			\node (b7) at (a7) {$(U_7,V_7)=(32\,x\,, 0)$};
            	\node (b8) at (a8) {$(U_8,V_8)=( -64\,, 0)$};
			\draw[<-,thick] (b0) -- (b1);
			\draw[<-,thick] (b1) -- (b2);
			\draw[<-,thick] (b2) -- (b3);
			\draw[<-,thick] (b3) -- (b4);
			\draw[<-,thick] (b4) -- (b5);
			\draw[<-,thick] (b5) -- (b6);
			\draw[<-,thick] (b6) -- (b7);
                \draw[<-,thick] (b7) -- (b8);
		\end{scope} 
	\end{tikzpicture}}}
\end{equation}
The configuration of the irreducible components of the anti-canonical divisor for the system is then given in terms of the classes of divisors $(\mathcal{H},\mathcal{E}_i)$ with $i=1, \dots, 9$, i.e.\ 
\begin{equation}\label{eq:syst_5_diagram}    \includegraphics[width=.43\textwidth,valign=c]{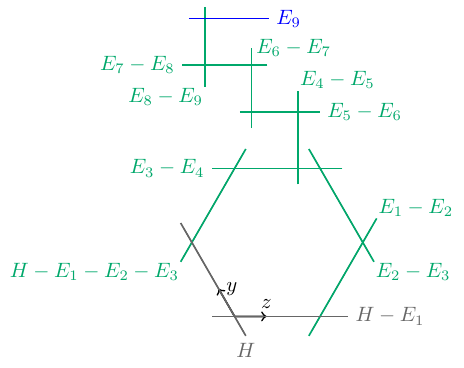} \qquad \hspace{10ex} 
    \begin{array}{l} 
       \delta_1 = \mathcal{H}-\mathcal{E}_1 - \mathcal{E}_2 - \mathcal{E}_3  \\[.5ex]
       \delta_2 =  \mathcal{E}_1 - \mathcal{E}_2  \\[.5ex]
       \delta_3 =  \mathcal{E}_2 - \mathcal{E}_3  \\[.5ex]
       \delta_4 =  \mathcal{E}_3 - \mathcal{E}_4    \\[.5ex]
       \delta_5 =   \mathcal{E}_4 - \mathcal{E}_5  \\[.5ex]
       \delta_6 =   \mathcal{E}_5 - \mathcal{E}_6   \\[.5ex]
       \delta_7 =   \mathcal{E}_6 - \mathcal{E}_7  \\[.5ex]
       \delta_8 =   \mathcal{E}_7 - \mathcal{E}_8    \\[.5ex]
       \delta_9 =   \mathcal{E}_8 - \mathcal{E}_9    \\[.5ex]
    \end{array}
\end{equation}
The diagram is determined starting from $\mathbb{P}^2$~\eqref{eq:CP2} and algorithmically proceeding by identifying the points at which blow-up the surface as stated in the previous section. The rational surface shown in~\eqref{eq:syst_5_diagram} will be the reference diagram for the other systems underpinning $\text{P}_{\text{I}}$, and in the following we will look for transformations of coordinates for the function variables $(z_5(x),y_5(x))$ satisfying the system V in terms of the coordinates satisfying the other systems.

\subsection{\texorpdfstring{System IX.B(2) and Bureau Hamiltonian $H_{\text{Bu92}}$}{H92}}
The system IX.B(2) is given by\footnote{See also equation (10) in \cite{Bureau table}.}
\begin{equation}\label{eq:syst_92}
\begin{cases} 
y' = -\,y^2 + z + 12\,q(x), \\[1ex]
z' = yz,
\end{cases} \qquad (q''(x) - 6\,q^2(x))'' = 0\,.
\end{equation}
We assume that the function $q(x)$ in the coefficient  satisfies the first Painlev\'e equation~\eqref{eq:P1}.  The system is non-Hamiltonian with respect to the standard 2-form $dy\wedge dz$ but changing the coefficients at $y^2$ in the first equation to $y^2/2$,   makes the system Hamiltonian with genus 0 Hamiltonian (or genus 1 if we consider an additional term depending on the independent variable $x$ but not on $y$ and $z$). We shall explore this case in a subsequent paper in more details.

The relation connecting the system IX.B(2) in the coordinates $(z_{92},y_{92})$ to system V in $(z_{5},y_{5})$ can be found in \cite[eq.\ (25)]{Bureau table}. Namely, one can easily check that the function  $z_{92}(x)$ satisfies the equation 
\begin{equation}\label{*}
z_{92}''=z_{92}^2+12 \,q\, z_{92}.
\end{equation}
Taking 
\begin{equation}\label{**}
z_{92}=6(y_5-q)\,, \qquad  y_{92}=\frac{z_{92}'}{z_{92}}=\frac{z_5-q'}{y_5-q}
\end{equation}
one can find that $$\frac{q''-y_5''}{q-y_5}=6(q+y_5)$$ by substituting~\eqref{**} into the second order differential equation for $z_{92}$~\eqref{*}. Hence, if both $y_5$ and $q$ are two different solutions of the first Painlev\'e equation~\eqref{eq:P1}, then the last equation   is trivial. The inverse transformation to~\eqref{**} is given by 
\vspace*{-1ex}
\begin{equation}\label{eq:change_5_92}
    y_5=\frac{z_{92}}{6}+q,\qquad z_5=\frac{y_{92}\,z_{92}}{6} + q'.
\end{equation}

We now consider the geometric approach via the regularisation of the system IX.B(2), for which we do not consider any assumption on the function $q(x)$. The fact that $q(x)$ satisfies $\text{P}_{\text{I}}$ is found as a condition for the system to be regularised along the last exceptional curves. 
The rational regularisation for the system IX.B(2) occurs by means of two cascades of blow-ups, i.e.\
\vspace*{-1ex}
\begin{equation}\label{eq:syst_92_cascades}
\vcenter{\hbox{
    \begin{tikzpicture}
		\begin{scope} 
			\path (0,0) coordinate (a0) -- ++(3.3,0)  coordinate (a1) -- ++(3.3,0)  coordinate (a2) -- ++(3.3,0)  coordinate (a3) -- ++(0,-1)  coordinate (a4)  -- ++(-4.3,0)  coordinate (a5) -- ++(0,-1)  coordinate (a6) -- ++(0,-1)  coordinate (a7);
			\node (b0) at (a0) {$(u_0,v_0)=(0\,, 0)$};
			\node (b1) at (a1) {$(U_1,V_1)=(0\,, 0)$};
			\node (b2) at (a2) {$(\tilde{u}_2,\tilde{v}_2)=(0\,,\tfrac{3}{2})$};
			\node (b3) at (a3) {$(u_3,v_3)=(0\,, 0)$};
			\node (b4) at (a4) {$(u_4,v_4)=(0\,,-18\,q(x))$};
			\node (b5) at (a5) {$(u_5,v_5)=(0\,,12\,q'(x))$};
			\node (b6) at (a6) {$(u_6,v_6)=(0\,,-12\,q''(x))$};
			\node (b7) at (a7) {$(u_7,v_7)=(0\,,24(q'''(x)-6\,q(x)q'(x)))$};
			\draw[<-,thick] (b0) -- (b1);
			\draw[<-,thick] (b1) -- (b2);
			\draw[<-,thick] (b2) -- (b3);
			\draw[<-,thick] (b3) -- (b4);
			\draw[<-,thick] (b4) -- (b5);
			\draw[<-,thick] (b5) -- (b6);
			\draw[<-,thick] (b6) -- (b7);
		\end{scope} 
        \begin{scope}
            \path (0,-4) coordinate (a0) -- ++(3.3,0)  coordinate (a1);
            \node (b0) at (a0) {$(U_0,V_0)=(0\,, 0)$};
		\node (b1) at (a1) {$(U_{8},V_{8})=(0\,, 0)$};
            \draw[<-,thick] (b0) -- (b1);
        \end{scope}
	\end{tikzpicture}}}
\end{equation}
and the configuration of inaccessible divisors is expressed in terms of the elements  $(\mathcal{K},\mathcal{F}_i)$ with $i=1, \dots, 10$. The configuration of $-2$-curves following the labelling in~\eqref{eq:diag_E8} is 
\begin{equation}\label{eq:syst_92_diag} \includegraphics[width=.45\textwidth,valign=c]{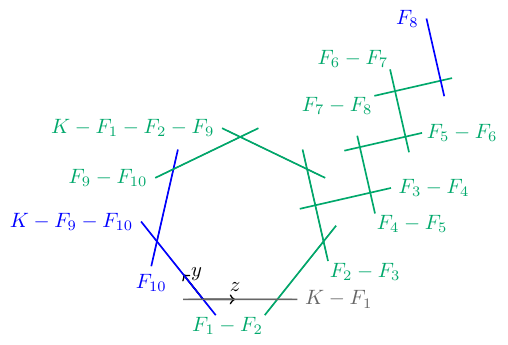} \qquad \hspace{10ex} 
    \begin{array}{l} 
       \delta_1 = \mathcal{F}_1 - \mathcal{F}_2  \\[.5ex]
       \delta_2 =  \mathcal{F}_9 - \mathcal{F}_{10}  \\[.5ex]
       \delta_3 =  \mathcal{K} - \mathcal{F}_1 - \mathcal{F}_2 - \mathcal{F}_9  \\[.5ex]
       \delta_4 =  \mathcal{F}_2 - \mathcal{F}_3    \\[.5ex]
       \delta_5 =  \mathcal{F}_3 - \mathcal{F}_4  \\[.5ex]
       \delta_6 =   \mathcal{F}_4 - \mathcal{F}_5   \\[.5ex]
       \delta_7 =   \mathcal{F}_5 - \mathcal{F}_6  \\[.5ex]
       \delta_8 =  \mathcal{F}_6 - \mathcal{F}_7    \\[.5ex]
       \delta_9 =   \mathcal{F}_7 - \mathcal{F}_8   \\[.5ex]
    \end{array}
\end{equation}
Following \cite{DFLS} in order to obtain a matching between the divisors for the two cases, we consider a further blow-up in the system V. This allows us to obtain the following correspondence: 
\begin{equation*}
    \begin{array}{l l}
        \mathcal{E}_1 = \mathcal{K} - \mathcal{F}_1 - \mathcal{F}_{10}  &   \mathcal{F}_1 = \mathcal{H} - \mathcal{E}_1 - \mathcal{E}_2  \\
        \mathcal{E}_2 = \mathcal{K} - \mathcal{F}_1 - \mathcal{F}_9 &    \mathcal{F}_2 =  \mathcal{E}_3\\
        \mathcal{E}_3 = \mathcal{F}_2 &  \mathcal{F}_3 = \mathcal{E}_4 \\
        \mathcal{E}_4 = \mathcal{F}_3 &  \mathcal{F}_4 = \mathcal{E}_5\\
        \mathcal{E}_5 = \mathcal{F}_4 & \mathcal{F}_5 = \mathcal{E}_{6}\\
        \mathcal{E}_6 = \mathcal{F}_5 & \mathcal{F}_6 = \mathcal{E}_7  \\
        \mathcal{E}_7 = \mathcal{F}_6 & \mathcal{F}_7 = \mathcal{E}_8 \\
        \mathcal{E}_8 = \mathcal{F}_{7} & \mathcal{F}_8 =\mathcal{E}_9 \\
        \mathcal{E}_9 = \mathcal{F}_{8}  & \mathcal{F}_9 = \mathcal{H} - \mathcal{E}_2 - \mathcal{E}_{10} \\
        \mathcal{E}_{10} = \mathcal{K} - \mathcal{F}_9 - \mathcal{F}_{10} & \mathcal{F}_{10} = \mathcal{H} - \mathcal{E}_1 - \mathcal{E}_{10} \\
        \mathcal{H} = 2\,\mathcal{K} - \mathcal{F}_1 - \mathcal{F}_9 - \mathcal{F}_{10} \hspace{10ex} & \mathcal{K} = 2\,\mathcal{H} - \mathcal{E}_1 - \mathcal{E}_2 - \mathcal{E}_{10}
    \end{array}
\end{equation*}
For the change of variables to express the coordinates $(z_5,y_5)$ in terms of the coordinates $(z_{92},y_{92})$ we consider the sets of points representing the coordinates axes for $\mathbb{A}_{(z_5,y_5)}$ in the system V: 
\begin{equation}\label{eq:ax_5_92}
      \begin{split}
        y_{5}&\colon \{z_{5}=0\}  \qquad \mathcal{H} - \mathcal{E}_1  = \mathcal{K} - \mathcal{F}_9\\ 
        z_{5}&\colon \{y_{5}=0\}  \qquad \mathcal{H} = 2\,\mathcal{K} - \mathcal{F}_1 - \mathcal{F}_9 - \mathcal{F}_{10} \\ 
    \end{split}
\end{equation}
The first relation suggests that the coordinate $y_5$ is a linear function in the variables $(z_{92},y_{92})$ intersecting the exceptional curve $\mathcal{F}_9$, i.e.\ we set 
\begin{equation}\label{eq:syst_92_intermediate}
    y_5 \colon a_{10}\,z_{92} + a_{01}\,y_{92}+a_{00} \,,
\end{equation}
with $a_{ij}=a_{ij}(x)$, and we require that this line goes through the point $(U_0,V_0)=(0,0)$ in~\eqref{eq:syst_92_cascades}. We rewrite~\eqref{eq:syst_92_intermediate} in the chart $(U_0,V_0)$, obtaining 
\begin{equation}
    a_{10}\,U_0 + a_{01}+a_{00}\,V_0 = 0\,,
\end{equation}
and imposing that the point $(U_0,V_0)=(0,0)$ is part of the line we find the condition $a_{01}=0$. The second expression in~\eqref{eq:ax_5_92} is a conic in the variables $(z_{92},y_{92})$, i.e.\
\begin{equation}\label{eq:z_5_conic_92}
    z_5 \colon b_{20}\,(z_{92})^2+b_{11}\,z_{92}\,y_{92} + b_{02}\,(y_{92})^2 + b_{10}\,z_{92} + b_{01}\,y_{92} + b_{00}=0 
\end{equation}
intersecting the exceptional divisors $\mathcal{F}_1$, $\mathcal{F}_9$ and $\mathcal{F}_{10}$, or going through the points with coordinates  
\begin{equation}
    (u_0,v_0) =(0,0)\,, \qquad (U_0,V_0) =(0,0)\,, \qquad (U_8,V_8) =(0,0)
\end{equation}
in~\eqref{eq:syst_92_cascades}. Imposing these constraints we fix the value of some of the coefficient functions $b_{ij}(x)$ in~\eqref{eq:z_5_conic_92} 
\begin{equation}
    b_{20} = 0 \,, \qquad b_{02} = 0 \,, \qquad b_{01} = 0\,. 
\end{equation}
At this stage the relations expressing the coordinates $(z_5,y_5)$ in terms of $(z_{92},y_{92})$ are
\begin{equation}
    y_5 = a_{10}\, z_{92}+a_{00} \,, \qquad z_5 = b_{11}\,z_{92}\,y_{92} + b_{10}\,z_{92}+b_{00}\,. 
\end{equation}
Lastly, requiring that the variables satisfy the original systems respectively, we find the change of coordinates~\eqref{eq:change_5_92}. In particular, we find the following constraints on the remaining coefficient functions $b_{ij}(x)$:
\begin{equation}
\begin{split} 
    &b_{00}=a_{00}'\,, \qquad b_{00}'=6\,a_{00}^2+x\,, \qquad b_{10}=a_{10}'\,, \qquad b_{11}=a_{10}\,, \qquad a_{10}=6\,a_{10}^2 \,, \qquad \\[1ex]
    &b_{11}\,q+b_{10}'=a_{00}\,a_{10}\,.
\end{split} 
\end{equation}
Solving these constraints we fix all the remaining coefficient functions, i.e. 
\begin{equation}
   a_{00}=q(x)\,, \qquad a_{10}=b_{11}=\frac{1}{6}\,, \qquad b_{10}=0\,, \qquad b_{00} = q'(x)  
\end{equation}
and inserting them in the general expressions~\eqref{eq:syst_92_intermediate} and~\eqref{eq:z_5_conic_92} we find the change of variables~\eqref{eq:change_5_92}. 

{
As previously mentioned, the system IX.B(2)~\eqref{eq:syst_92} is non-Hamiltonian with respect to the standard canonical $2$-form $dz_{92}\wedge dy_{92}$, while the system V~\eqref{eq:syst_5} is Hamiltonian with respect to $dz_{5}\wedge dy_{5}$. Hence, we can ask whether the system IX.B(2) is Hamiltonian with respect to a non-standard version of the $2$-form given by the pull back of the $2$-form related to the change of variables~\eqref{eq:change_5_92}. We have 
\begin{equation}\label{eq:2form_92}
    \begin{split}
        dz_5 \wedge dy_5 &= \left( \frac{\partial y_5}{\partial z_{92}}\,\frac{\partial z_5}{\partial y_{92}} - \frac{\partial y_5}{\partial y_{92}}\,\frac{\partial z_5}{\partial z_{92}} \right) dz_{92} \wedge dy_{92} = \frac{z_{92}}{36}\,dz_{92} \wedge dy_{92} \,, 
    \end{split}
\end{equation}
and we look for the function $H_{\text{Bu92}}(z_{92},y_{92})$ such that the system IX.B(2)~\eqref{eq:syst_92} admits the following expression: 
\begin{equation}
    \begin{cases}
        y_{92}'= \dfrac{36}{z_{92}}\,\dfrac{\partial H_{\text{Bu92}}}{\partial z_{92}}\,, \\[2ex]
        z_{92}'= -\dfrac{36}{z_{92}}\,\dfrac{\partial H_{\text{Bu92}}}{\partial y_{92}}\,. 
    \end{cases}
\end{equation}
The function $H_{\text{Bu92}}$ exists and reads as
\vspace*{-4ex}
\begin{equation}\label{eq:Ham_Bu92}
    H_{\text{Bu92}}(z,y) = z^2 \left( \frac{z}{108} - \frac{y^2}{72}+ \frac{q}{6} \right) +f(x)\,, \hspace{10ex} \includegraphics[width=.17\textwidth,valign=c]{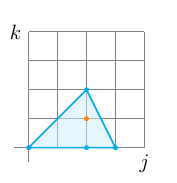}   
\end{equation}
\vspace*{-2ex}

\noindent
with $f$ a function only depending on $x$. The corresponding Newton polygon represents the algebraic curve with genus $1$ algebraic and area equal to $3$. The system IX.B(2) is then Hamiltonian with respect to the non-standard $2$-form~\eqref{eq:2form_92}. We emphasise that in \eqref{eq:Ham_Bu92} if $f(x)=0$ the corresponding Newton polygon is of genus 0.

}

\subsection{System IX.B(5)}
The system IX.B(5) is given by\footnote{Note the typo in \cite{Guillot} in the second equation, see equation (12) in \cite{Bureau table} instead. } 
\begin{equation}
\begin{cases} \label{eq:syst_95}
y' = -\,y^2 + z + 3\,q(x), \\[1ex]
z' = 4\,yz -  9\,q'(x),
\end{cases} \qquad (q''(x) - 6\,q^2(x))'' = 0\,.
\end{equation}
We assume that  function $q(x)$ in the coefficient satisfies the first Painlev\'e equation~\eqref{eq:P1}.  The system is non-Hamiltonian with respect to the standard symplectic 2-form {$dz \wedge dy$}. However, it is worth noticing that changing the coefficients of~$y^2$ to $2y^2$ in the first equation in~\eqref{eq:syst_92}, the resulting system is Hamiltonian with genus 1 Hamiltonian
\vspace*{-3ex}
\begin{equation}
H_{\text{95m}}(z,y)=9 q'\,y+3 q\, z+2\, z y^2+\frac{z^2}{2}\hspace{10ex} \includegraphics[width=.17\textwidth,valign=c]{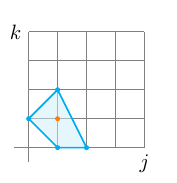}   
\end{equation}
where the labelling 95m stays for the modified version of the system IX.B(5).

The regularisation of the system IX.B(5) takes place through the following cascades of blow-ups from the compactified space $\mathbb{P}^2$: 
\begin{equation}\label{eq:syst_95_cascades}
\vcenter{\hbox{
    \begin{tikzpicture}
    \begin{scope} 
        \path (0,0) coordinate (a0) -- ++(3,0)  coordinate (a1) -- ++(3,0)  coordinate (a2) -- ++(3,0)  coordinate (a3) -- ++(0,-1)  coordinate (a4);
        \node (b0) at (a0) {$(u_0,v_0)=(0\,, 0)$};
        \node (b1) at (a1) {$(U_1,V_1)=(0\,, 0)$};
        \node (b2) at (a2) {$(\tilde{u}_2,\tilde{v}_2)=(0\,, 3)$};
        \node (b3) at (a3) {$(u_3,v_3)=(0\,, 0)$};
        \node (b4) at (a4) {$(u_4,v_4)=(0\,, -9\,q(x))$};
        \draw[<-,thick] (b0) -- (b1);
        \draw[<-,thick] (b1) -- (b2);
        \draw[<-,thick] (b2) -- (b3);
        \draw[<-,thick] (b3) -- (b4);
        \end{scope} 
    \begin{scope} 
        \path (0,-2) coordinate (a0) -- ++(3,0)  coordinate (a1) -- ++(3.5,0)  coordinate (a2) -- ++(0,-1)  coordinate (a3) -- ++(0,-1)  coordinate (a4) -- ++(-4.6,0)  coordinate (a5);
        \node (b0) at (a0) {$(U_0,V_0)=(0\,, 0)$};
        \node (b1) at (a1) {$(U_5,V_5)=(0\,, 0)$};
        \node (b2) at (a2) {$(U_6,V_6)=(3\,q'(x),0)$};
        \node (b3) at (a3) {$(U_{7},V_{7})=(\frac{3}{2}\,q''(x),0)$};
        \node (b4) at (a4) {$(U_{8},V_{8})=(\frac{3}{2}\,q'''(x)-9\,q(x)q'(x),0)$};
        \node (b5) at (a5) {$(u_{9},v_{9})=(0\,,0)$};
        \draw[<-,thick] (b0) -- (b1);
        \draw[<-,thick] (b1) -- (b2);
        \draw[<-,thick] (b2) -- (b3);
        \draw[<-,thick] (b3) -- (b4);
        \draw[<-,thick] (b4) -- (b5);
        \end{scope} 
    \end{tikzpicture}}}
\end{equation}
with the regularisation condition for the system at the last blow-up being the one reported in~\eqref{eq:syst_95}, i.e.\ if we remove the assumption on $q(x)$ from the beginning, we find that the system is regularised if $q(x)$ satisfies $\text{P}_{\text{I}}$ in~\eqref{eq:P1}. The configuration of the intersection diagram and the corresponding identification of the components of the anti-canonical divisor are:
\begin{equation}\label{eq:syst_95_diag} \includegraphics[width=.45\textwidth,valign=c]{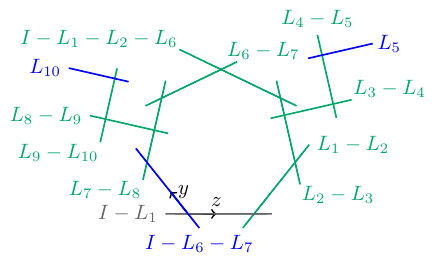} \qquad \hspace{10ex} 
    \begin{array}{l} 
       \delta_1 = \mathcal{L}_1 - \mathcal{L}_2  \\[.5ex]
       \delta_2 =  \mathcal{L}_4 - \mathcal{L}_{5}  \\[.5ex]
       \delta_3 =  \mathcal{L}_3 - \mathcal{L}_4  \\[.5ex]
       \delta_4 = \mathcal{L}_2 - \mathcal{L}_3   \\[.5ex]
       \delta_5 =  \mathcal{I} - \mathcal{L}_1 - \mathcal{L}_2 - \mathcal{L}_6  \\[.5ex]
       \delta_6 =   \mathcal{L}_6 - \mathcal{L}_7   \\[.5ex]
       \delta_7 =   \mathcal{L}_7 - \mathcal{L}_8  \\[.5ex]
       \delta_8 =  \mathcal{L}_8 - \mathcal{L}_{9}    \\[.5ex]
       \delta_9 =   \mathcal{L}_{9} - \mathcal{L}_{10}   \\[.5ex]
    \end{array}
\end{equation}

\begin{theorem}
Let $(y_{95}\,,z_{95})$ satisfy the system IX.B(5) in~\eqref{eq:syst_95} and $(y_{92}\,,z_{92})$ satisfy the system IX.B(2) in~\eqref{eq:syst_92} with the same $q(x)$ satisfying the first Painlev\'e equation~\eqref{eq:P1}. Then the systems are related by the birational transformation of variables
\begin{equation}\label{eq:coord_95_92}
y_{95}=-\frac{y_{92}}{2},\qquad z_{95}=-\frac{z_{92}}{2}+\frac{3 y_{92}^2}{4}-9q
\end{equation} 
and
\begin{equation}\label{eq:coord_92_95}
y_{92}=-2y_{95},\qquad z_{92}=-2(z_{95}-3y_{95}^2+9q).
\end{equation} 
\end{theorem}
\begin{proof}
We prove  the theorem with both the geometric approach and the iterative polynomial regularisation. Comparing~\eqref{eq:syst_92_diag} and~\eqref{eq:syst_95_diag} we establish the correspondence between the two classes of divisors, namely $(\mathcal{K},\mathcal{F}_i)$ and $(\mathcal{I},\mathcal{L}_i)$. We obtain the following equivalences
\begin{equation*}
    \begin{array}{l l}
        \mathcal{F}_1 = \mathcal{I} - \mathcal{L}_2 - \mathcal{L}_3  &   \mathcal{L}_1 = \mathcal{K} - \mathcal{F}_2 - \mathcal{F}_3  \\
        \mathcal{F}_2 = \mathcal{I} - \mathcal{L}_1 - \mathcal{L}_3 &    \mathcal{L}_2 = \mathcal{K} - \mathcal{F}_1 - \mathcal{F}_3\\
        \mathcal{F}_3 = \mathcal{I} - \mathcal{L}_1 - \mathcal{L}_2 &  \mathcal{L}_3 = \mathcal{K} - \mathcal{F}_1 - \mathcal{F}_2\\
        \mathcal{F}_4 = \mathcal{L}_6 &  \mathcal{L}_4 = \mathcal{F}_9\\
        \mathcal{F}_5 = \mathcal{L}_7 & \mathcal{L}_5 = \mathcal{F}_{10}\\
        \mathcal{F}_6 = \mathcal{L}_8 & \mathcal{L}_6 = \mathcal{F}_4  \\
        \mathcal{F}_7 = \mathcal{L}_9 & \mathcal{L}_7 = \mathcal{F}_5 \\
        \mathcal{F}_8 = \mathcal{L}_{10} & \mathcal{L}_8 =\mathcal{F}_6 \\
        \mathcal{F}_9 = \mathcal{L}_4 & \mathcal{L}_9 = \mathcal{F}_7\\
        \mathcal{F}_{10} = \mathcal{L}_5 & \mathcal{L}_{10} = \mathcal{F}_8\\
        \mathcal{K} = 2\,\mathcal{I} - \mathcal{L}_1 - \mathcal{L}_2 - \mathcal{L}_3 \hspace{10ex} & \mathcal{I} = 2\,\mathcal{K} - \mathcal{F}_1 - \mathcal{F}_2 - \mathcal{F}_3
    \end{array}
\end{equation*}
In order to find the changes of variables we consider the expressions for the axes $(z_{92}\,,y_{92})$, i.e.\ 
\begin{equation}\label{eq:change_92_95}
    \begin{split}
        y_{92}&\colon \{z_{92}=0\}  \qquad \mathcal{K} - \mathcal{F}_1  = \mathcal{I} - \mathcal{L}_1\\ 
        z_{92}&\colon \{y_{92}=0\}  \qquad \mathcal{K} - \mathcal{F}_9 - \mathcal{F}_{10}   = 2\,\mathcal{I} - \mathcal{L}_1- \mathcal{L}_2- \mathcal{L}_3- \mathcal{L}_4- \mathcal{L}_5\\ 
    \end{split}
\end{equation}
and those for the axes $(z_{95}\,,y_{95})$
\begin{equation}\label{eq:change_95_92}
    \begin{split}
        y_{95}&\colon \{z_{95}=0\}  \qquad  \mathcal{I} - \mathcal{L}_1 = \mathcal{K} - \mathcal{F}_1 \\ 
        z_{95}&\colon \{y_{95}=0\}  \qquad \mathcal{I} - \mathcal{L}_6 - \mathcal{L}_7   = 2\,\mathcal{K} - \mathcal{F}_1- \mathcal{F}_2- \mathcal{F}_3- \mathcal{F}_4- \mathcal{F}_5.\\ 
    \end{split}
\end{equation}
We start by considering the relations in~\eqref{eq:change_92_95}. The expression for $y_{92}$ is given in terms of a linear function in $(z_{95},y_{95})$ intersecting the exceptional line $\mathcal{L}_1$ or, in other words, going through the point with coordinates $(u_0,v_0)=(0,0)$ in~\eqref{eq:syst_95_cascades}. This allows us to write the following 
\begin{equation}\label{eq:line_y92}
    y_{92}\colon  a_{10}\,y_{95} + a_{01}\,z_{95} + a_{00}=0\,.
\end{equation}
By inverting the changes of coordinates in~\eqref{eq:CP2}, the relation~\eqref{eq:line_y92} can be rewritten in the chart~$(u_0,v_0)$ as 
\begin{equation}
    a_{10}\,v_0 + a_{01} + a_{00}\,u_0 = 0\,, 
\end{equation}
and imposing that the point $(u_0,v_0)=(0,0)$ is part of the line, we find the condition $a_{01}=0$. For the expression for the axis $z_{92}$ the procedure is similar, but involves more steps since the expression in~\eqref{eq:change_92_95} is a conic 
\begin{equation}\label{eq:line_z92}
    z_{92}\colon b_{20}\,(y_{95})^2+b_{11}\,y_{95}\,z_{95} + b_{02}\,(z_{95})^2 + b_{10}\,y_{95} + b_{01}\,z_{95} + b_{00}=0 
\end{equation}
and to fix the parameters of the conic $b_{ij}(x)$ we impose that the conic intersects the lines $\mathcal{L}_1$, $\mathcal{L}_2$, $\mathcal{L}_3$, $\mathcal{L}_4$ and $\mathcal{L}_5$, or in other words it goes through the points with coordinates identified by the charts~$(u_0,v_0)$, $\dots$, $(u_4,v_4)$ in~\eqref{eq:syst_95_cascades}. We start by rewriting~\eqref{eq:line_z92} in terms of the chart $(u_0,v_0)$ by inverting the relations in~\eqref{eq:CP2}, i.e.
\begin{equation}
    b_{20}\,(v_0)^2+b_{11}\,v_0 + b_{02} + b_{10}\,u_0\,v_0 + b_{01}\,u_0 + b_{00}\,(u_0)^2=0 \,,
\end{equation}
and imposing the passage through the point $(u_0,v_0)=(0,0)$ we fix the parameter $b_{02}=0$. We now consider the same expression in the chart $(U_1,V_1)$, i.e. 
\begin{equation}
    b_{20}\,V_1+b_{11} + b_{10}\,U_1\,V_1 + b_{01}\,U_1 + b_{00}\,(U_1)^2\,V_1=0 \,,
\end{equation}
from which we find $b_{11}=0$ imposing the passage through the point $(U_1,V_1)=(0,0)$. Considering the same for the remaining intersecting points we fix the coefficient functions
\begin{equation}
    b_{20} = -3\,b_{01}\,, \qquad b_{10} = 0 \,, \qquad b_{00} = 9q(x)\,b_{01}\,.
\end{equation}
Therefore, we can rewrite the expressions~\eqref{eq:line_y92} and~\eqref{eq:line_z92} as 
\begin{equation}
    \begin{split}
    y_{92} = a_{10}\,y_{95}+a_{00} , \qquad 
    z_{92} = b_{01}(z_{95} - 3\,y_{95}^2+9q(x)).
    \end{split} 
\end{equation}
Finally, the last coefficients are fixed by requiring that the coordinates $(y_{92},z_{92})$ satisfy the system~\eqref{eq:syst_92}, and we obtain the final expressions~\eqref{eq:coord_95_92}. 

By inverting the form of~\eqref{eq:coord_95_92} we find~\eqref{eq:coord_92_95}, that is also verified starting from~\eqref{eq:change_95_92} and proceeding analogously to what we shown above, i.e.\ looking for the constraints to obtain the set of coefficient functions for a linear and a conic expression to describe the axes $y_{95}$ and $z_{95}$ respectively.  

We give a second proof of the theorem via the iterative polynomial regularisation. We consider the system in $(z_{95},y_{95})$ and proceed with the blow-up transformations until it is regularised. In the regularisation process we stress that we make use of the intermediate change of variables $(u_i,v_i)\mapsto (\tilde{u}_i,1/\tilde{v}_i)$ in correspondence with the first non-zero point of indeterminacy. We are interested in one of the two regularised system in the final charts. In coordinates $(u,v)$ this takes the form
\begin{equation}\label{eq:syst_95_reg_1}
   \begin{cases} 
   u'=-2+6 q u^2 -u^3 v,\\[1ex]
   v'=-6 q- u v+u^2 v^2\,. 
   \end{cases}
\end{equation}
The regularising transformation for this cascade is obtained by tracking back all the coordinate transformations to get the regularised system. Hence, we can reconstruct the expression of the coordinate functions $(z_{95},y_{95})$ in terms of the variables of the final chart $(u,v)$, as 
\begin{equation}
   z_{95}=\frac{3+u^2(uv-9q)}{u^2}\,, \qquad  y_{95}=\frac{1}{u}\,. 
\end{equation}
We now consider the system~\eqref{eq:syst_95_reg_1} in the variables $(u,v)$ and iterate the regularisation process. We choose a compactification for $(u,v) \in \mathbb{C}^2$ and identify the points of indeterminacy. Just after one blow-up we obtain the regular polynomial system in the new set of coordinates $(w(x),t(x))$ 
\begin{equation}\label{eq:syst_95_reg_2}
   \begin{cases} 
   w'=2\,w^2+t-6q,\\[1ex] 
   t'=-2\,wt,
   \end{cases}
\end{equation}
with $(u,v)\mapsto (1/w,\,wt)$. We observe that the system~\eqref{eq:syst_95_reg_2} is similar to the system IX.B(2) in~\eqref{eq:syst_92}, and they coincide once we consider the rescaling of the variables
\begin{equation}
    w=-\frac{y_{92}}{2},\qquad t=-\frac{z_{92}}{2}\,.
\end{equation}
Taking the composition of all the transformation of variables and their rescaling, we obtain the statement of the theorem. 
\end{proof}

A direct consequence of the previous Theorem is the following

\begin{corollary}
Combining the transformations relating the system V and the system IX.B(2) in~\eqref{eq:change_5_92} with the transformations~\eqref{eq:change_95_92} relating the system IX.B(2) and IX.B(5), we have 
\begin{equation}\label{eq:syst_95_5}
    y_{95}=-\frac{z_5-q'}{2(y_5-q)},\qquad z_{95}=-9 q-3(y_5-q)+\frac{3(z_5-q')^2}{4(y_5-q)^2}
\end{equation}
and 
\begin{equation}\label{eq:syst_5_95}
    y_5=-2q+y_{95}^2-\frac{z_{95}}{3},\qquad z_5=6 q\, y_{95}-2y_{95}^3+\frac{2}{3}\,y_{95}\,z_{95}+q'.
\end{equation}
\end{corollary}
The expressions in~\eqref{eq:syst_95_5} and~\eqref{eq:syst_5_95} can be obtained from the geometric approach as well, from the comparison between~\eqref{eq:syst_5_diagram} and~\eqref{eq:syst_92_diag}. This allows us to find the correspondence of the components:  
\begin{equation*}
    \begin{array}{l l}
        \mathcal{E}_1 = \mathcal{I} - \mathcal{L}_1 - \mathcal{L}_5  &   \mathcal{L}_1 = 2\,\mathcal{H} - \mathcal{E}_1 - \mathcal{E}_2 - \mathcal{E}_3 - \mathcal{E}_4  -\mathcal{E}_{10}   \\
        \mathcal{E}_2 = \mathcal{I} - \mathcal{L}_1 - \mathcal{L}_4 &    \mathcal{L}_2 = \mathcal{H} - \mathcal{E}_4  -\mathcal{E}_{10} \\
        \mathcal{E}_3 = \mathcal{I} - \mathcal{L}_1 - \mathcal{L}_3 &  \mathcal{L}_3 = \mathcal{H} - \mathcal{E}_3  -\mathcal{E}_{10}  \\
        \mathcal{E}_4 = \mathcal{I} - \mathcal{L}_1 - \mathcal{L}_2 &  \mathcal{L}_4 = \mathcal{H} - \mathcal{E}_2  -\mathcal{E}_{10}\\
        \mathcal{E}_5 = \mathcal{L}_6 & \mathcal{L}_5 = \mathcal{H} - \mathcal{E}_1  -\mathcal{E}_{10}\\
        \mathcal{E}_6 = \mathcal{L}_7 & \mathcal{L}_6 = \mathcal{E}_5  \\
        \mathcal{E}_7 = \mathcal{L}_8 & \mathcal{L}_7 = \mathcal{E}_6 \\
        \mathcal{E}_8 = \mathcal{L}_{9} & \mathcal{L}_8 =\mathcal{E}_7 \\
        \mathcal{E}_9 = \mathcal{L}_{10} & \mathcal{L}_9 = \mathcal{E}_8\\
        \mathcal{E}_{10} = 2\,\mathcal{I} - \mathcal{L}_1 - \mathcal{L}_2 - \mathcal{L}_3 -2\,\mathcal{L}_5 & \mathcal{L}_{10} = \mathcal{E}_9\\
        \mathcal{H} = 3\,\mathcal{I} - 2\,\mathcal{L}_1 - \mathcal{L}_2 - \mathcal{L}_3 - \mathcal{L}_4 - \mathcal{L}_5 \hspace{10ex} & \mathcal{I} = 3\,\mathcal{H} - \mathcal{E}_1 - \mathcal{E}_2 - \mathcal{E}_3 - \mathcal{E}_4 - 2\, \mathcal{E}_{10}
    \end{array}
\end{equation*}
Analogously to the way we proceeded above, we consider the axes for the coordinates $(z_5,y_5)$ in terms of the coordinates $(z_{95},y_{95})$, i.e.
\vspace*{-1ex}
\begin{equation}
    \begin{split}
        y_5 &\colon \{z_5=0\} \qquad \mathcal{H}-\mathcal{E}_1 = 2\, \mathcal{I} - \mathcal{L}_1 - \mathcal{L}_2 - \mathcal{L}_3 - \mathcal{L}_4 \\[1ex]
        z_5 &\colon \{y_5=0\} \qquad \mathcal{H} = 3\,\mathcal{I} - 2\,\mathcal{L}_1 - \mathcal{L}_2 - \mathcal{L}_3 - \mathcal{L}_4 - \mathcal{L}_5.
    \end{split}
\end{equation}
The first expression gives us a conic for $y_5$ as a function of $(z_{95},y_{95})$, whereas the second expression describes a cubic for $z_5$ in the variables $(z_{95},y_{95})$. The generic conic and cubic curves are 
\begin{align}
\label{eq:y5_in_95}
    y_5 &\colon~~ b_{20}\,(z_{95})^2+b_{11}\,z_{95}\,y_{95} + b_{02}\,(y_{95})^2 + b_{10}\,z_{95} + b_{01}\,y_{95} + b_{00}=0\,,  \\[1ex]
    \label{eq:z5_in_95}
   \begin{split}  z_5 &\colon~~ c_{30}\,(z_{95})^3+c_{21}\,z_{95}^2\,y_{95} +c_{12}\,z_{95}\,y_{95}^2 + c_{03}\,(y_{95})^3  \\ 
   & ~~~~ + c_{20}\,(z_{95})^2+c_{11}\,z_{95}\,y_{95} + c_{02}\,(y_{95})^2 + c_{10}\,z_{95} + c_{01}\,y_{95} + c_{00}=0 \,. 
   \end{split}
\end{align}
The conic~\eqref{eq:y5_in_95} goes through the points $(u_0,v_0)=(0,0)$, $\dots$, $(u_3,v_3)=(0,0)$ in~\eqref{eq:syst_95_cascades}, fixing the coefficients 
\begin{equation}
    b_{20}=0\,, \qquad b_{11}=0\,, \qquad b_{10}=-\frac{b_{02}}{3}\,, \qquad b_{01}=0\,. 
\end{equation}
The cubic~\eqref{eq:z5_in_95} goes through the points $(u_0,v_0)=(0,0)$, $(U_1,V_1)=(0,0)$, $(\tilde{u}_2,\tilde{v}_2)=(0,3)$, $(u_3,v_3)=(0,0)$ and $(u_4,v_4)=(0,-9\,q(x))$ in~\eqref{eq:syst_95_cascades}, producing the constraints on the coefficients
\begin{equation}
    c_{30}=0\,, \quad c_{21}=0\,, \quad c_{20}=-\frac{c_{12}}{3}\,, \quad c_{11}=-\frac{c_{03}}{3}\,, \quad c_{10}=-\frac{c_{02}}{3}+ 27\,c_{12}\,q(x).
\end{equation}
The change of variables for $y_5$ and $z_5$ is then 
\begin{equation}\label{eq:intermediate_change_5_95}
\begin{split} 
    y_5&=b_{02} \left( y_{95}^2 -\frac{z_{95}}{3} \right) +b_{00} \,, \\[1ex] 
    z_5 &= c_{00}+c_{01} y_{95}+c_{02} \left( y_{95}^2-\frac{ z_{95}}{3}\right)+c_{03} \left( y_{95}^3-\frac{y_{95} z_{95}}{3} \right)+c_{12} \left(y_{95}^2 z_{95}-\frac{z_{95}^2}{3}  +27 q z_{95}\right). 
    \end{split} 
\end{equation}
Lastly, the remaining coefficient functions are completely determined by requiring that the transformed variables still satisfy the system V. In particular, equating the coefficient functions of terms of the same order in $(z_{95},y_{95})$, we obtain the following constraints
\begin{equation}
\begin{split} 
    c_{00}&= b_{00}'+ 3\,b_{02}\,q'\,,\quad c_{01}=6 q\,,\quad c_{02}= b_{02}' \,, \quad c_{03}=-2\,b_{02}\,,\quad c_{12}=0\,, \\[1ex]
    b_{00}&= -2q\,, \qquad b_{02}=1 \,, 
\end{split} 
\end{equation}
which included in the expression~\eqref{eq:intermediate_change_5_95} yield the transformation~\eqref{eq:syst_5_95}, with the function $q(x)$ always satisfying $\text{P}_{\text{I}}$.

\subsection{System XIV} 
The system XIV is given by 
\begin{equation}\label{eq:syst_14}
\begin{array}{c c}
\begin{cases} 
y' = y(2z - y) + 3\,p(x)\,y + r(x), \\[1ex]
z' = z(y - z) - 2\,p(x)\,z,
\end{cases}  & \qquad \begin{aligned}
     q(x) &= \frac{1}{12}(p'(x) + p^2(x) - r(x)), \\[1ex]
      r'(x) &= p(x)\,r(x),\quad   (q''(x) - 6\,q^2(x))'' = 0\,. 
\end{aligned}
\end{array}
\end{equation} 
We assume that the function $q(x)$, which is hidden  in the regularisation conditions, satisfies $\text{P}_{\text{I}}$ in~\eqref{eq:P1}, and we stress the fact that this system is non-Hamiltonian. 

To regularise the system XIV, we find three cascades of blow-ups from $\mathbb{P}^2$, emerging from the origins of the charts $(u_0,v_0)$ and $(U_0,V_0)$ and from the additional point of indeterminacy visible in both charts, with coordinates $(u_0,v_0)=(0,3/2)$ and $(U_0,V_0)=(2/3,0)$ respectively:
\begin{equation}\label{eq:syst_14_cascade}
\vcenter{\hbox{
    \begin{tikzpicture}
    \begin{scope} 
        \path (0,0) coordinate (a0) -- ++(3,0)  coordinate (a1) -- ++(3.4,0)  coordinate (a2);
        \node (b0) at (a0) {$(u_0,v_0)=(0\,, 0)$};
        \node (b1) at (a1) {$(u_1,v_1)=(0\,, 0)$};
        \node (b2) at (a2) {$(u_2,v_2)=(0\,, -r(x))$};
        \draw[<-,thick] (b0) -- (b1);
        \draw[<-,thick] (b1) -- (b2);
        \end{scope} 
    \begin{scope} 
        \path (0,-1) coordinate (a0) -- ++(3.4,0)  coordinate (a1) -- ++(4.6,0)  coordinate (a2) -- ++(-1,0) node (A) {${\,}_{\,_{\,}}$} -- ++(0,-1)  coordinate (a3) -- ++(0,-1.7)  coordinate (a4)  -- ++(0,-2.7)  coordinate (a5);
        \node (b0) at (a0) {$(u_3,v_3)=(0\,, \tfrac{3}{2})$};
        \node (b1) at (a1) {$(u_4,v_4)=(0\,, 3\,p(x))$};
        \node (b2) at (a2) {$(u_5,v_5)=(0\,, \tfrac{1}{2} \left(r(x)-3 p'(x)\right))$};
        \node (b3) at (a3) {$ \hspace*{-30ex} (u_6,v_6)=(0\,, p''(x)+2\, p(x) p'(x)-\frac{2}{3} p(x) r(x)-\tfrac{1}{3}r'(x))=(0\,, 12\, q'(x))$};
        \node (b4) at (a4) {$ \hspace*{-24ex}\begin{aligned}  (u_7,v_7)&=(0\,, \tfrac{2}{3}\,p^2(x)r(x)-2\,p^2(x)p'(x)+\tfrac{2}{3}\,p'(x)r(x)-2(p'(x)^2+p(x)r'(x) \\ &\hspace{6ex}-3\,p(x)p''(x)+\tfrac{1}{3}\,r''(x)-p'''(x))\\
        &=(0\,,-12\,q''(x)-p(x)(p'(x)-p^2(x))'\,)
        \end{aligned} $};
        \node[] (b5) at (a5) {$\hspace*{-18ex}\begin{aligned}  (u_8,v_8)&=(0\,, 2\, p^{(IV)}(x)+6\, p'''(x) p(x)-\tfrac{1}{3}\, r(x) p''(x)+4\, p^2(x) p''(x)\\
        &\hspace{6ex}-\tfrac{7}{3}\, p'(x) r'(x)+\tfrac{4}{3}\, p(x) r(x) p'(x)+2\, p(x) (p'(x))^2+11\, p'(x) p''(x)\\
        &\hspace{6ex}-2\, p(x) r''(x)-\tfrac{4}{3} \, p^2(x) r'(x)-\tfrac{2}{3}\, p(x) r^2(x)-\tfrac{2}{3}\, r'''(x)-\tfrac{1}{3}\, r(x) r'(x)\,)\\
        &=(0\,,\, 12(2(q''(x)-3q(x)^2)'+p(x)q'(x)+2p(x)q''(x))\,)
        \end{aligned} $};
        \draw[<-,thick] (b0) -- (b1);
        \draw[<-,thick] (b1) -- (b2);
        \draw[<-,thick] (A) -- (b3);
        \draw[<-,thick] (b3) -- (b4);
        \draw[<-,thick] (b4) -- (b5);
        \end{scope} 
    \begin{scope} 
        \path (0,-8.3) coordinate (a0) -- ++(3.2,0)  coordinate (a1);
        \node (b0) at (a0) {$(U_0,V_0)=(0\,, 0)$};
        \node (b1) at (a1) {$(U_{10},V_{10})=(0\,, 0)$};
        \draw[<-,thick] (b0) -- (b1);
    \end{scope} 
    \end{tikzpicture}}}
\end{equation}
At each step in the cascades, we rewrite the expression for the coordinates in terms of the functions $q(x)$ and $p(x)$ to obtain a more concise expression. 

The configuration of the inaccessible components of the anti-canonical divisor takes the following form in terms of the elements $(\mathcal{J},\mathcal{G}_{i})$, with $i=1, \dots, 11$
\begin{equation}\label{eq:syst_14_diagram}
    \includegraphics[width=.34\textwidth,valign=c]{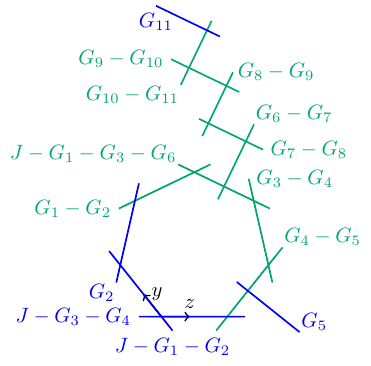} \hspace{10ex} 
    \begin{array}{l} 
       \delta_1 = \mathcal{G}_1 - \mathcal{G}_2  \\[.5ex]
       \delta_2 = \mathcal{G}_4 - \mathcal{G}_5  \\[.5ex]
       \delta_3 =  \mathcal{G}_3 - \mathcal{G}_4  \\[.5ex]
       \delta_4 =  \mathcal{J}-\mathcal{G}_1 - \mathcal{G}_2 - \mathcal{G}_6    \\[.5ex]
       \delta_5 =   \mathcal{G}_6 - \mathcal{G}_7 \\[.5ex]
       \delta_6 =  \mathcal{G}_7 - \mathcal{G}_8 \\[.5ex]
       \delta_7 =  \mathcal{G}_8 - \mathcal{G}_9 \\[.5ex]
       \delta_8 =   \mathcal{G}_9 - \mathcal{G}_{10}  \\[.5ex]
       \delta_9 =   \mathcal{G}_{10} - \mathcal{G}_{11}  \\[.5ex]
    \end{array}
\end{equation}

\noindent 
The relation to system~\eqref{eq:syst_5} can be easily obtained using information in \cite[Section 3.11]{Guillot}. Namely, we have the following birational transformations:
\begin{equation}
    y_{14}=\frac{(6 q+r-6y_5)(q-y_5)}{qp-py_{5}+z_5-q'},\qquad z_{14}=\frac{z_5-q'}{y_5-q}{-p},
\end{equation}
and the inverse one:
\begin{equation}\label{eq:syst_5_14}
    y_5=\frac{1}{6}(6 q+r+y_{14}z_{14}),\qquad z_5=\frac{1}{6}(r z_{14}+y_{14}z_{14}^2 +6q'+p(r+y_{14}z_{14}) ).
\end{equation}
These relations are obtained from the following transformations. First taking $z=u-p$  in system~\eqref{eq:syst_14} one can derive a second order differential equation for the function $u$ which has the form $$u''=-u u' +u^3-12 q u+12 q' $$
with $q $ defined as in~\eqref{eq:syst_14}. Then taking 
\vspace*{-2ex}
\begin{equation}
u=\frac{z_{5}-q'}{y_5-q} 
\end{equation}
yields the transformations above between the systems~\eqref{eq:syst_14} and~\eqref{eq:syst_5}.

To proceed with the geometric approach and compare the irreducible components associated with the system V~\eqref{eq:syst_5_diagram} and the system XIV in~\eqref{eq:syst_14_diagram}, we consider two more blow-ups in system V, so that we deal with the same number of elements. We get the following equivalences:
\begin{equation*}
    \begin{array}{l l}
        \mathcal{E}_1 = \mathcal{J} - \mathcal{G}_1 - \mathcal{G}_5  &   \mathcal{G}_1 = 2\,\mathcal{H} - \mathcal{E}_1 - \mathcal{E}_2 - \mathcal{E}_3  -\mathcal{E}_{10} - \mathcal{E}_{11}  \\
        \mathcal{E}_2 = \mathcal{J} - \mathcal{G}_1 - \mathcal{G}_4 &    \mathcal{G}_2 = \mathcal{H} - \mathcal{E}_{10} - \mathcal{E}_{11} \\
        \mathcal{E}_3 = \mathcal{J} - \mathcal{G}_1 - \mathcal{G}_3 &  \mathcal{G}_3 = \mathcal{H} - \mathcal{E}_3  -\mathcal{E}_{10}  \\
        \mathcal{E}_4 = \mathcal{G}_6 &  \mathcal{G}_4 = \mathcal{H} - \mathcal{E}_2  -\mathcal{E}_{10}\\
        \mathcal{E}_5 = \mathcal{G}_7 & \mathcal{G}_5 = \mathcal{H} - \mathcal{E}_1  -\mathcal{E}_{10}\\
        \mathcal{E}_6 = \mathcal{G}_8 & \mathcal{G}_6 = \mathcal{E}_4  \\
        \mathcal{E}_7 = \mathcal{G}_9 & \mathcal{G}_7 = \mathcal{E}_5 \\
        \mathcal{E}_8 = \mathcal{G}_{10} & \mathcal{G}_8 =\mathcal{E}_6 \\
        \mathcal{E}_9 = \mathcal{G}_{11} & \mathcal{G}_9 = \mathcal{E}_7\\
        \mathcal{E}_{10} = 2\,\mathcal{J} - \mathcal{G}_1 - \mathcal{G}_2 - \mathcal{G}_3- \mathcal{G}_4 -\mathcal{G}_5 & \mathcal{G}_{10} = \mathcal{E}_8\\
        \mathcal{E}_{11} = \mathcal{J} - \mathcal{G}_1 - \mathcal{G}_2 \hspace{10ex} & \mathcal{G}_{11} = \mathcal{E}_{9}\\
        \mathcal{H} = 3\,\mathcal{J} - 2\,\mathcal{G}_1 - \mathcal{G}_2 - \mathcal{G}_3 - \mathcal{G}_4 - \mathcal{G}_5 \hspace{10ex} & \mathcal{J} = 3\,\mathcal{H} - \mathcal{E}_1 - \mathcal{E}_2 - \mathcal{E}_3 - 2\, \mathcal{E}_{10} - \mathcal{E}_{11} .
    \end{array}
\end{equation*}
We look for the expressions for $(z_5,y_5)$ in terms of the coordinates $(z_{14},y_{14})$, by considering the equations for the coordinate axes in the system V 
\begin{equation}\label{eq:y5z5_syst14}
    \begin{split}
        y_5 &\colon \{z_5=0\} \qquad \mathcal{H}-\mathcal{E}_1 = 2\, \mathcal{J} - \mathcal{G}_1 - \mathcal{G}_2 - \mathcal{G}_3 - \mathcal{G}_4 \\[1ex]
        z_5 &\colon \{y_5=0\} \qquad \mathcal{H} = 3\,\mathcal{J} - 2\,\mathcal{G}_1 - \mathcal{G}_2 - \mathcal{G}_3 - \mathcal{G}_4 - \mathcal{G}_5
    \end{split}
\end{equation}
giving rise to a conic for $y_5$ passing through the points $(U_0,V_0)=(0,0)$, $(U_{10},V_{10})=(0,0)$, $(u_0,v_0)=(0,0)$ and $(u_1,v_1)=(0,0)$ in~\eqref{eq:syst_14_cascade}, and a cubic for $z_5$ going through these same points and the additional point $(u_2,v_2)=(0,-r)$. 
The constraints on the coefficient functions are  
\begin{equation}
    b_{20}=0\,, \quad b_{02}=0\,, \quad b_{10}=0\,, \quad b_{01}=0\,, 
\end{equation}
for the conic representing $y_5$ in the first expression of~\eqref{eq:y5z5_syst14}, and 
\begin{equation}
    c_{30}=0\,, \quad c_{03}=0\,, \quad c_{20}=0\,, \quad c_{02}=0\,, \quad c_{10}=r\,c_{21}\,, 
\end{equation}
for the cubic representing $z_5$ in the second expression of~\eqref{eq:y5z5_syst14}. 
The change of variables from the coordinates $(z_5,y_5)$ to $(z_{14},y_{14})$ is then given by the relations
\begin{equation}
\begin{split} 
    y_5 &= b_{11}\, z_{14}\,y_{14}+b_{00}\,,  \\[1ex]
    z_5 &= c_{00}+c_{01}\, y_{14}+c_{11}\, y_{14} z_{14}+c_{12}\, y_{14}^2 z_{14}+c_{21} (r z_{14}+ y_{14} z_{14}^2)\,. 
    \end{split}
\end{equation}
By requiring that the systems V and XIV are satisfied, we set the remaining coefficients. In particular, we have 
\vspace*{-2ex}
\begin{equation}
    \begin{split}
        c_{00}&= b_{00}'\,, \quad c_{01}=0\,, \quad c_{11}=b_{11}\,p+b_{11}'\,, \quad c_{12}=0\,, \quad c_{21}=b_{11}   \\[1ex]
        b_{00}&=\frac{1}{12}(p'+p^2+r) \,, \quad b_{11}=\frac{1}{6}\,, 
    \end{split}
\end{equation}
and rewriting $b_{00}$ in terms of $q(x)$ as defined in~\eqref{eq:syst_14} we reconstruct the transformation~\eqref{eq:syst_5_14}.  

{Interestingly, the pair of functions  $(p(x),r(x))$ appearing in the system XIV~\eqref{eq:syst_14} satisfy the system IX.B(2)~\eqref{eq:syst_92}, as first observed in~\cite{ECGF}. 
}

\section{\texorpdfstring{Bureau-Guillot systems related to $\text{P}_{\text{II}}$}{P2}}\label{sec:P2}
 Analogously to what we have observed in the previous section, the systems related to $\text{P}_{\text{II}}$ admit a surface type that is represented as the extended Dynkin diagram $E_7^{(1)}$, here depicted with a specific labelling for the irreducible components of the anti-canonical divisor
 \vspace*{-2ex}
\begin{equation}\label{eq:diag_E7}
    \includegraphics[width=.34\textwidth,valign=c]{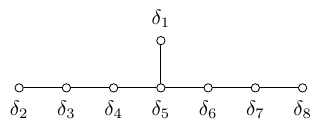}
\end{equation}

\subsection{System IX.B(3)}
The system IX.B(3) is given by 
\begin{equation}\label{eq:syst_93}
\begin{cases} 
y' = -y^2 + z - \dfrac{f(x)}{2}, \\[1ex]
z' = 2\,yz + \dfrac{1}{2}(2\alpha+1),
\end{cases} \qquad f''(x) = 0\,, \quad \alpha\in \mathbb{C} \text{ is a  constant}\,.
\end{equation}
The function $y$ solves the second Painlev\'e equation~\eqref{eq:P2} 
provided that $f(x)=x$.  We will consider this system as the  reference system for the other systems underpinning $\text{P}_{\text{II}}$. 
The system~IX.B(3) is Hamiltonian with the Okamoto Hamiltonian of genus 1
\vspace*{-3ex}
\begin{equation}
    H_{\text{Ok}}(z,y)=\frac{1}{2} z^2 - \left(y ^2 + \frac{x}{2}\right) z  - \left(\alpha + \frac{1}{2}\right) y\,\,, \hspace{10ex} \includegraphics[width=.17\textwidth,valign=c]{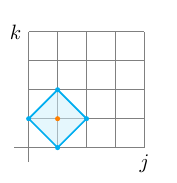} 
\end{equation}
The regularisation of the Hamiltonian system IX.B(3) with an arbitrary $f(x)$ is realised via two cascades of blow-ups from~$\mathbb{P}^2$ in~\eqref{eq:CP2} starting from indeterminacies in the origins of both the affine charts $(u_0,v_0)$ and $(U_0,V_0)$: 
\begin{equation}
\vcenter{\hbox{
	\begin{tikzpicture}
		\begin{scope} 
			\path (0,0) coordinate (a0) -- ++(3.2,0)  coordinate (a1) -- ++(3.2,0)  coordinate (a2) -- ++(3.2,0)  coordinate (a3) -- ++(0,-1)  coordinate (a4)  -- ++(-4.5,0)  coordinate (a5);
			\node (b0) at (a0) {$(u_0,v_0)=(0\,, 0)$};
			\node (b1) at (a1) {$(U_1,V_1)=(0\,, 0)$};
			\node (b2) at (a2) {$(\tilde{u}_2,\tilde{v}_2)=(0\,, 2)$};
			\node (b3) at (a3) {$(u_3,v_3)=(0\,, 0)$};
			\node (b4) at (a4) {$(u_4,v_4)=(0\,, f(x))$};
			\node (b5) at (a5) {$(u_5,v_5)=(0\,, \alpha + \tfrac{1}{2}+f'(x))$};
			\draw[<-,thick] (b0) -- (b1);
			\draw[<-,thick] (b1) -- (b2);
			\draw[<-,thick] (b2) -- (b3);
			\draw[<-,thick] (b3) -- (b4);
			\draw[<-,thick] (b4) -- (b5);
		\end{scope} 
		\begin{scope} 
			\path (0,-2) coordinate (a0) -- ++(3.2,0)  coordinate (a1) -- ++(3.8,0)  coordinate (a2);
			\node (b0) at (a0) {$(U_0,V_0)=(0\,, 0)$};
			\node (b1) at (a1) {$(U_6,V_6)=(0\,, 0)$};
			\node (b2) at (a2) {$(U_7,V_7)=(-(\alpha + \tfrac{1}{2}),0)$};
			\draw[<-,thick] (b0) -- (b1);
			\draw[<-,thick] (b1) -- (b2);
		\end{scope} 
	\end{tikzpicture}}}
\end{equation}
for which we find again the condition given in~\eqref{eq:syst_93} for the function $f(x)$, that is at most linear in the independent variable $x$. Taking track of all the transformations, we can associate the minimal diagram of the form~\eqref{eq:diag_E7} {with} the system with the following expressions for the components: 
\begin{equation}\label{eq:syst_93_diagram}
    \includegraphics[width=.45\textwidth,valign=c]{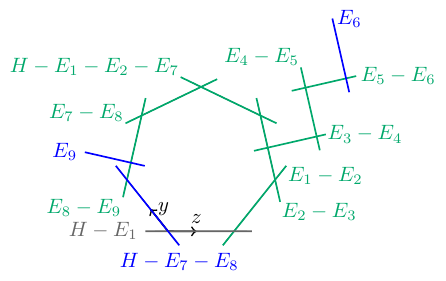} \hspace{10ex} 
    \begin{array}{l} 
       \delta_1 = \mathcal{E}_1 - \mathcal{E}_2  \\[.5ex]
       \delta_2 =  \mathcal{E}_8 - \mathcal{E}_{9}  \\[.5ex]
       \delta_3 =  \mathcal{E}_7 - \mathcal{E}_8  \\[.5ex]
       \delta_4 = \mathcal{H}- \mathcal{E}_1 - \mathcal{E}_2 - \mathcal{E}_7   \\[.5ex]
       \delta_5 =    \mathcal{E}_2 - \mathcal{E}_3  \\[.5ex]
       \delta_6 =   \mathcal{E}_3 - \mathcal{E}_4   \\[.5ex]
       \delta_7 =    \mathcal{E}_4 - \mathcal{E}_5 \\[.5ex]
       \delta_8 =   \mathcal{E}_5 - \mathcal{E}_6    \\[.5ex]
    \end{array}
\end{equation}
We notice that in this case there is an ambiguity in the choice of the components $\delta_2$ and $\delta_8$ because of the symmetry of the diagram $E_7^{(1)}$. 


\subsection{\texorpdfstring{System XIII and Bureau Hamiltonians $H_{\text{Bu13}}$}{H13}}

The system XIII is given by\footnote{Note the typo in the resonance conditions in the current version of \cite{Guillot}.}
\begin{equation}\label{eq:syst_13}
\begin{array}{c c}
\hspace*{-3ex}\begin{cases} 
y' = \frac{1}{2}\,y(2\,z - y) + 2\,p(x)\,y ,\\[1ex]
z' = \frac{1}{2}\,z(3\,y - 2\,z) - 4\,p(x)\,z + 2\,p^2(x) - 2\,p'(x) + f,
\end{cases}   \begin{aligned}
    (p''(x) -  2\,p^3(x) - f(x)\,p(x))' &= 0, \\[1ex]
      f''(x) &= 0\,. 
\end{aligned}
\end{array}
\end{equation} 
The coefficient function $p(x)$ satisfies $\text{P}_{\text{II}}$, and the system is non-Hamiltonian. 

For the regularisation of XIII we find three cascades of blow-ups for the surface, emerging from the origins of the charts $(u_0,v_0)$, $(U_0,V_0)$ and from the additional point visible from both charts, with coordinates $(u_0,v_0)=(0,1)$ and $(U_0,V_0)=(1,0)$ respectively:
\begin{equation}\label{eq:syst_13_cascades}
\vcenter{\hbox{
    \begin{tikzpicture}
    \begin{scope} 
        \path (0,0) coordinate (a0) -- ++(3.1,0)  coordinate (a1);
        \node (b0) at (a0) {$(u_0,v_0)=(0\,, 0)$};
        \node (b1) at (a1) {$(u_1,v_1)=(0\,, 0)$};
        \draw[<-,thick] (b0) -- (b1);
        \end{scope} 
    \begin{scope} 
        \path (0,-1) coordinate (a0) -- ++(3.5,0)  coordinate (a1) -- ++(5.5,0)  coordinate (a2) -- ++(0,-1)  coordinate (a3);
        \node (b0) at (a0) {$(u_0,v_0)=(0\,, 1)$};
        \node (b1) at (a1) {$(u_2,v_2)=(0\,, 4\,p(x))$};
        \node (b2) at (a2) {$(u_3,v_3)=(0\,,-(f(x)+2\,p'(x)+2\,p^2(x)))$};
        \node (b3) at (a3) {$\hspace*{-20ex}(u_{4},v_{4})=(0\,,f'(x)+4\,p(x)p'(x)+2\,p''(x))$};
        \draw[<-,thick] (b0) -- (b1);
        \draw[<-,thick] (b1) -- (b2);
        \draw[<-,thick] (b2) -- (b3);
        \end{scope} 
    \begin{scope} 
        \path (0,-3) coordinate (a0) -- ++(3.2,0)  coordinate (a1) -- ++(5,0)  coordinate (a2) -- ++(0,-1)  coordinate (a3);
        \node (b0) at (a0) {$(U_0,V_0)=(0\,, 0)$};
        \node (b1) at (a1) {$(U_5,V_5)=(0\,, 0)$};
        \node (b2) at (a2) {$(U_6,V_6)=(2\,p'(x)-f(x)-2\,p^2(x),0)$};
        \node (b3) at (a3) {$\hspace*{-20ex}(U_{7},V_{7})=(-2 (f'(x)+2 f(x) p(x)-2 p''(x)+4 p^3(x) )\,,0)$};
        \draw[<-,thick] (b0) -- (b1);
        \draw[<-,thick] (b1) -- (b2);
        \draw[<-,thick] (b2) -- (b3);
    \end{scope} 
    \end{tikzpicture}}}
\end{equation}
The conditions for the system to be regularised at the end of the second and third cascades of blow-ups are 
\vspace*{-2ex}
\begin{align}
2\,pf'+2fp'+12\,p^2 p'+f''-2\,p'''&=0\,,\\[1ex]
-2\,pf'-2fp'-12\,p^2 p'+f''+2\,p'''&=0\,,
\end{align}
and their sum and difference give 
\vspace*{-2ex}
\begin{align}
f''&=0\,,\\[1ex]
pf'+fp'+6\,p^2p'-p'''&=0\,,
\end{align}
respectively. The last identity  can be easily integrated once, and we find the conditions stated in~\eqref{eq:syst_13} for $f(x)$ and $p(x)$.

We construct the specific configuration of the irreducible components of the inaccessible divisor, and we find a surface type with Dynkin diagram $E_7^{(1)}$ in~\eqref{eq:diag_E7}, i.e.\
\vspace*{-1ex}
\begin{equation}\label{eq:syst_13_diagram}
    \includegraphics[width=.47\textwidth,valign=c]{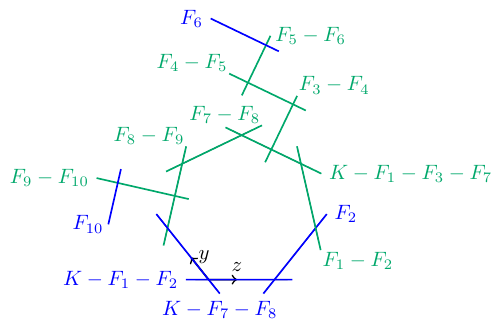} \hspace{10ex} 
    \begin{array}{l} 
       \delta_1 = \mathcal{F}_1 - \mathcal{F}_2  \\[.5ex]
       \delta_2 =  \mathcal{F}_9 - \mathcal{F}_{10}  \\[.5ex]
       \delta_3 =  \mathcal{F}_8 - \mathcal{F}_{9}  \\[.5ex]
       \delta_4 =  \mathcal{F}_7 - \mathcal{F}_{8}   \\[.5ex]
       \delta_5 =    \mathcal{K}- \mathcal{F}_1 - \mathcal{F}_3 - \mathcal{F}_7  \\[.5ex]
       \delta_6 =   \mathcal{F}_3 - \mathcal{F}_4   \\[.5ex]
       \delta_7 =    \mathcal{F}_4 - \mathcal{F}_5 \\[.5ex]
       \delta_8 =   \mathcal{F}_5 - \mathcal{F}_6    \\[.5ex]
    \end{array}
\end{equation}
\begin{theorem}\label{thm:syst_13_93}
Let $(z_{13},y_{13})$ satisfy the system XIII in~\eqref{eq:syst_13} -- with $f(x)=x$ and $p(x)$  $\text{P}_{\text{II}}$ transcendent -- and $(z_{13},y_{13})$ satisfy the system IX.B(3) in\eqref{eq:syst_93}. Then, the systems are related by the birational transformation
\vspace*{-2ex}
\begin{equation}\label{eq:syst_13_93}
    y_{13}=2p-2y_{93},\qquad z_{13}=\frac{2(x+2\,y_{93}^2-z_{93})+2p'-2p^2-x}{2p-2y_{93}}
\end{equation}

\vspace*{-1ex}

\noindent 
and
\begin{equation}\label{eq:syst_93_13}
y_{93}=\frac{1}{2}(2p-y_{13}),\qquad z_{93}=\frac{1}{2}(y_{13}^2 -y_{13}\,z_{13}-4p\, y_{13}+x+2p^2+2p').
\end{equation}

\end{theorem}

\begin{proof}
    We prove the result via the geometric approach and the iterative polynomial regularisation. 

In order to compare the two intersection diagrams in~\eqref{eq:syst_93_diagram} and~\eqref{eq:syst_13_diagram} we need to operate with an additional blow-up in the system IX.B(3) so that we can compare the same number of elements. We establish the following correspondence 
\begin{equation*}
    \begin{array}{l l}
        \mathcal{E}_1 = \mathcal{K} - \mathcal{F}_2 - \mathcal{F}_7  &   \mathcal{F}_1 = \mathcal{H} - \mathcal{E}_2  -\mathcal{E}_{10}  \\
        \mathcal{E}_2 =  \mathcal{K} - \mathcal{F}_1 - \mathcal{F}_7 &    \mathcal{F}_2 = \mathcal{H} - \mathcal{E}_1  -\mathcal{E}_{10}\\
        \mathcal{E}_3 = \mathcal{F}_3 &  \mathcal{F}_3 = \mathcal{E}_3   \\
        \mathcal{E}_4 = \mathcal{F}_4 &  \mathcal{F}_4 = \mathcal{E}_4 \\
        \mathcal{E}_5 = \mathcal{F}_5 & \mathcal{F}_5 = \mathcal{E}_5 \\
        \mathcal{E}_6 = \mathcal{F}_6 & \mathcal{F}_6 = \mathcal{E}_6  \\
        \mathcal{E}_7 = \mathcal{F}_8 & \mathcal{F}_7 = \mathcal{H} - \mathcal{E}_1  -\mathcal{E}_{2} \\
        \mathcal{E}_8 = \mathcal{F}_{9} & \mathcal{F}_8 =\mathcal{E}_7 \\
        \mathcal{E}_9 = \mathcal{F}_{10} & \mathcal{F}_9 = \mathcal{E}_8\\
        \mathcal{E}_{10} = \mathcal{K} - \mathcal{F}_1 - \mathcal{F}_2 &  \mathcal{F}_{10} = \mathcal{E}_9 \\
        \mathcal{H} = 2\,\mathcal{K} - \mathcal{F}_1 - \mathcal{F}_2 - \mathcal{F}_7 \hspace{10ex} & \mathcal{K} = 2\,\mathcal{H} - \mathcal{E}_1 - \mathcal{E}_2 -\mathcal{E}_{10} 
    \end{array}
\end{equation*}
The expressions for the axes $y_{93}$ and $z_{93}$ are
\begin{equation}
\begin{split} 
    y_{93} &\colon \{z_{93}=0\} \qquad \mathcal{H}-\mathcal{E}_1 = \mathcal{K} - \mathcal{F}_1\\[1ex]
    z_{93} &\colon \{y_{93}=0\} \qquad \mathcal{H}-\mathcal{E}_7 - \mathcal{E}_8 = 2\,\mathcal{K} - \mathcal{F}_1 - \mathcal{F}_2 - \mathcal{F}_7 - \mathcal{F}_8 - \mathcal{F}_9 
\end{split} 
\end{equation}
hence a linear function for $y_{93}$ going through the point $(u_0,v_0)=(0,0)$ in~\eqref{eq:syst_13_cascades} and a conic for $z_{93}$ going through the points with coordinates $(u_0,v_0)=(0,0)$, $(u_1,v_1)=(0,0)$, $(u_0,v_0)=(0,1)$, $(u_2,v_2)=(0,4\,p)$ and $(u_3,v_3)=(0,x-2\,p'-2\,p^2)$ in~\eqref{eq:syst_13_cascades}. Imposing the constraints, we fix some of the coefficient functions in the generic expressions for the transformation of coordinates, i.e.
\begin{equation}\label{eq:interm_13_93}
    \begin{split}
        y_{93}&= a_{01}\,y_{13} + a_{00} \,,\\[1ex]
        z_{93}&= b_{11}\left( -y_{13}^2+y_{13}\,z_{13} +4\,p\,y_{13} -x-2p'-2p^2 \right)\,.
    \end{split}
\end{equation}
Lastly, requiring that the transformed variables satisfy the system IX.B(3) we can specify the remaining coefficient functions, being
\vspace*{-1ex}
\begin{equation}
    a_{01} = b_{11}\,, \qquad  2\,a_{01}(a_{00}+p)+a_{01}'=4 \,b_{11}\,p \,, \qquad b_{11}=-\frac{1}{2}\,,
\end{equation}
and substituting these into the~\eqref{eq:interm_13_93} we obtain~\eqref{eq:syst_93_13}. 

We now consider the iterative regularisation for the system XIII, by first tracking back the changes of variables for the regularised system in the most suitable chart. Identifying with $(u,v)$ the variables of the final chart, we establish the regularising birational transformation 
\vspace*{-1ex}
\begin{equation}\label{eq:syst_13_iter}
    y_{13} = \frac{1}{u}\,, \qquad z_{13} = u\! \left(2p'-2p^2-f+u(u v-4fp-8p^3-2f'+4p'')\right)\,.
\end{equation}
We promote the system in the final affine chart $(u,v)$ to be the new system to be regularised. This system is regularised after two blow-ups, with final system in the variables $(w,t)$ being
\vspace*{-1ex}
\begin{equation}\label{eq:interm_wt_93}
    \begin{cases}
        w'=\frac{1}{2}(4p w+2\,t-w^2+4\,p'-4\,p^2-2f)\,,\\[1ex]
        t'=wt-2pt+f'-2p''-2fp+4\,p^3\,,
    \end{cases}
\end{equation}
and same resonance condition as the previous step acting on $f$ and $p$. The birational transformation connecting the system in $(u,v)$ and the system in $(w,t)$ is 
\begin{equation}
 u=\frac{1}{w}\,,\qquad v=w(4fp+8\,p^3+wt+2f'-4\,p'').
\end{equation}
Combining the two birational transformations we get 
\begin{equation}\label{eq:Bu_transf_1}
y_{13}=w,\qquad z_{13}=\frac{t+2p'-2p^2-f}{w}\,,
\end{equation}
and 
\begin{equation}\label{eq:Bu_transf_2}
w=y_{13},\qquad t=f+2p^2+ y_{13}\, z_{13} -2p'.
\end{equation}
Finally, assuming that $f(x)=x$ and $p(x)$ satisfies $\text{P}_{\text{II}}$ in~\eqref{eq:P2} with parameter $\alpha$, taking  
\begin{equation}
    w=2p-2y_{93},\qquad t=2(x+2y_{93}^2-z_{93})
\end{equation}
yields the birational transformation~\eqref{eq:syst_13_93} and after the inversion~\eqref{eq:syst_93_13}. 
\end{proof}

Thanks to the evident symmetry of the diagram for $E_7^{(1)}$ it is possible to adopt a different choice for the matching of the components of the inaccessible divisor when comparing~\eqref{eq:syst_13_diagram} with~\eqref{eq:syst_93_diagram} by fixing $\mathcal{E}_6=\mathcal{F}_{10}$ instead of $\mathcal{E}_6=\mathcal{F}_6$. 

In this case, we get the following relations 
\begin{equation}
\begin{split} 
    y_{93} &\colon \{z_{93}=0\} \qquad \mathcal{H}-\mathcal{E}_1 = \mathcal{K} - \mathcal{F}_1\\[1ex]
    z_{93} &\colon \{y_{93}=0\} \qquad \mathcal{H}-\mathcal{E}_7 - \mathcal{E}_8 = 2\,\mathcal{K} - \mathcal{F}_1 - \mathcal{F}_2 - \mathcal{F}_3 - \mathcal{F}_4 - \mathcal{F}_5 
\end{split} 
\end{equation}
with the conic going through the points $(u_0,v_0)=(0,0)$, $(u_1,v_1)=(0,0)$, $(U_0,V_0)=(0,0)$, $(U_5,V_5)=(0,0)$ and $(U_6,V_6)=(2\,p'(x)-f(x)-2\,p^2(x),0)$ in~\eqref{eq:syst_13_cascades}. After imposing all the constraints, the birational transformation between the coordinates takes the form:
\begin{equation}\label{eq:systm_93_13_bis1}
    y_{93} = \frac{y_{13}}{2}-p \,, \qquad z_{93}= \frac{1}{2}(y_{13}\,z_{13} +x+2p^2-2\,p')\,, 
\end{equation}
and inverse transformation
\begin{equation}\label{eq:systm_93_13_bis2}
    y_{13}= 2( y_{93}- p)\,,\qquad z_{13}=  \frac{2 z_{93}-x-2 p^2+2 p'}{2(y_{93}-p)}\,. 
\end{equation}
The transformations \eqref{eq:systm_93_13_bis1} and \eqref{eq:systm_93_13_bis2} hold provided that the function $p(x)$ satisfies $\text{P}_{\text{II}}$ in~\eqref{eq:P2} with $\alpha \mapsto - \alpha$. Alternatively, if we assume that $p(x)$ satisfies $\text{P}_{\text{II}}$ in~\eqref{eq:P2}, the formulas hold if $\alpha=0$.

It is worth noticing that the expressions for the birational transformations in Theorem~\ref{thm:syst_13_93} are still valid if the function $p(x)$ satisfies the Riccati equation, or it is a rational solution for some specific value of the parameter $\alpha$ in $\text{P}_{\text{II}}$. 

{
As mentioned above, the system XIII is non-Hamiltonian with respect to the standard $2$-form $dz_{13} \wedge dy_{13}$, however, we can determine whether it is Hamiltonian with respect to the pull back induced by the change of variables~\eqref{eq:systm_93_13_bis1}. We consider the $2$-form 
\begin{equation}
    dz_{93} \wedge dy_{93} = \left( \frac{\partial y_{93}}{\partial z_{13}}\,\frac{\partial z_{93}}{\partial y_{13}} - \frac{\partial y_{93}}{\partial y_{13}}\,\frac{\partial z_{93}}{\partial z_{13}} \right) dz_{13} \wedge dy_{13} = -\frac{y_{13}}{2}\,dz_{13} \wedge dy_{13} \,, 
\end{equation}
and look for a function $H_{\text{Bu13}}$ such that the system XIII~\eqref{eq:syst_13} can be written in the following form
\begin{equation}
    \begin{cases}
         y_{13}'= -\dfrac{2}{y_{13}}\,\dfrac{\partial H_{\text{Bu13}}}{\partial z_{13}}\,, \\[2ex]
        z_{13}'= \dfrac{2}{y_{13}}\,\dfrac{\partial H_{\text{Bu13}}}{\partial y_{13}}\,. 
    \end{cases}
\end{equation}
The function $H_{\text{Bu13}}$ exists and it reads as 
\vspace*{-4ex}
\begin{equation}\label{eq:Ham_Bu13}
    H_{\text{Bu13}}(z,y) = \frac{y^2}{2}  \left(p^2-p'-p\, z +\frac{y\, z - z^2+x}{2}  \right)+ g(x)
\hspace{10ex} \includegraphics[width=.17\textwidth,valign=c]{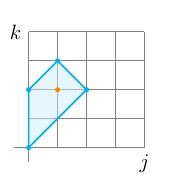} 
\end{equation}
\vspace*{-3ex}

\noindent 
with $g$ function of $x$ only. The Hamiltonian function is a genus $1$ algebraic curve and its associated Newton polygon has an area equal to $3$. With the function $g(x)=0$, the Newton polygon would have a diamond shape and the corresponding area equal to $2$.  
}

Finally, we introduce a new Hamiltonian system birationally equivalent to XIII.  After the first regularisation process, the systems in both the final affine charts is Hamiltonian, and the algebraic curve representing the Hamiltonian is of genus 2. Following the proof of the Theorem~\ref{thm:syst_13_93}, the iteration of the regularisation procedure yields the system~\eqref{eq:interm_wt_93}. This system is Hamiltonian and the algebraic curve is of genus 1 
\vspace*{-4ex}
\begin{equation}\label{eq:HamBu}
\begin{aligned} 
    H_{\text{Bu13}}^1(w,t)&=- w \left(4 p^3+2 f p-2 p''+f'\right)+2 p \,w t \\[1ex]
    &~~~~ -t\left(2 p^2 +2 p' +f + \frac{1}{2}\right) -
   \frac{w^2 t}{2}+\frac{t^2}{2}
   \end{aligned} 
   \hspace{10ex} \includegraphics[width=.17\textwidth,valign=c]{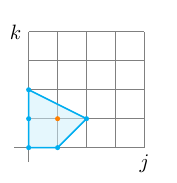} 
\end{equation}
\vspace*{-2ex}

\noindent
We label this function Bu13 to emphasise the connection to the system XIII, and introduce the additional upper index $1$.

The following holds true 
\begin{theorem}
The system XIII is birationally equivalent to the system with the Hamiltonian~\eqref{eq:HamBu} of genus 1 by the birational transformations~\eqref{eq:Bu_transf_1} and~\eqref{eq:Bu_transf_2}. 
\end{theorem}
\begin{proof}
    The proof is given by iterative regularisation as in Theorem~\ref{thm:syst_13_93}.
\end{proof}
The shape of the Newton polygon for the Hamiltonian $H_{\text{Bu13}}^1$ suggests the following change of variables in system~\eqref{eq:interm_wt_93}: $w\to 2p-2w$ and $t\to t+f-2\alpha x.$ Assuming that $p(x)$ satisfies $p''=2p^3+f(x)p+\alpha$ we obtain a quadratic system
 \begin{equation}
 \begin{cases}
     w'=\alpha x+w^2 - \dfrac{t}{2}\,,\\[1ex]
    t'=-2w(t+f(x)-2\alpha x)\,,
 \end{cases}
 \end{equation} 
which is Hamiltonian with the Hamiltonian 
\vspace*{-4ex}
\begin{equation} 
H_{\text{Bu13}}^2(w,t)=-2\alpha x\, w^2 +f(x)w^2 +\alpha x\, t+w^2 t-\frac{t^2}{4}\hspace{10ex} \includegraphics[width=.17\textwidth,valign=c]{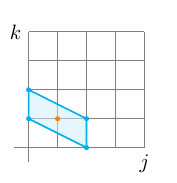} 
\end{equation}
\vspace*{-3ex}

\noindent
and which gives the second Painlev\'e equation with parameter $\alpha$ and $f(x)=x$ for the function $w$. The Newton polygon for the Hamiltonian $H_{\text{Bu13}}^2$ is parallelogram-shaped  which is different from the Newton polygon of the  Okamoto Hamiltonian. However, they both have same genus (equal to 1) and same area (equal to 2).

\section{Discussion and open questions}\label{sec:disc}

The classification problem of systems sharing a particular property is generally a very complicated task. In this paper, we showed how to find birational transformations between certain Bureau-Guillot systems from~\cite{Guillot} by using the geometric approach and the iterative polynomial regularisation. This solves the so-called Painlev\'e equivalence problem for the Bureau-Guillot systems underpinning $\text{P}_{\text{I}}$ and $\text{P}_{\text{II}}$ and substantially complements and provides a justification for the results in \cite{Guillot}. There are several questions that arise after our analysis, possible object of future works. 

In Section~\ref{sec:P1} we examined  the selected systems which are all related to $\text{P}_{\text{I}}$. It might be interesting to deepen the study of the elliptic version of the systems as well, and we expect to find a geometric description analogous to that analysed in this paper. 

{
In a forthcoming paper we propose the study of new systems where other Painlev\'e transcendents (from $\text{P}_{\text{III}}$ to $\text{P}_{\text{VI}}$) appear at the level of the coefficients, and we conjecture that it should be possible to construct analogous systems for the quasi-Painlev\'e case, which maybe can arise in systems with higher degree polynomials or rational functions. 
}

We briefly mentioned the use of the Newton polygons for the polynomial Hamiltonians, without going into too much detail. This construction seems to carry very useful auxiliary information about the systems and possibly it may play a crucial role while solving the Painlev\'e or quasi-Painlev\'e equivalence problem.  We examined this in the paper \cite{new} for certain $\text{P}_{\text{IV}}$-type equations. Most of the systems that we considered in this work are non-Hamiltonian {with respect to the standard $2$-form}, and it is then natural to ask whether it is possible to introduce Newton polytopes for general systems in a similar way as the Newton polygons are introduced for Hamiltonian systems. For generic non-Hamiltonian systems one can try to associate a Newton polytope in a different way. For instance, one can reduce a system to a scalar equation for one of the variables, or find the degrees of functions of all terms in the right-hand side and then take the convex hull of the Minkowski sum. Another possibility is similar to what is proposed in~\cite{Chiba}, adapted to our case since the coefficient functions are not rational in $x$: for a monomial $y^n z^{m}$ in the first equation one associates the point $(n-1,m)$ and  the point $(n,m-1)$ in the second equation and then taking the convex hull. The genus will then be possibly defined as discussed in \cite{MDTK2}. Such a definition corroborates that for the Newton polygons if the system is in addition Hamiltonian. The Newton polytopes so defined should carry useful information and be useful in the classification problem of generic polynomial systems, which is in general very difficult due to many subcases to consider. However, as we can see for the systems associated with $\text{P}_{\text{I}}$ in this paper, we observe different genera so determined, and so the question of the minimal genus and minimal area of the Newton polytope for the Painlev\'e equations (or birationally equivalent systems in general) arises. For instance, from systems IX.B(2)~\eqref{eq:syst_92} and XIV~\eqref{eq:syst_14} we have genus 0 for a Newton polytope associated with the system using the modified Chiba's definition, with area equal to 1. For the Hamiltonian system~\eqref{eq:syst_5} we have genus 1 and area equal to 2. For $\text{P}_{\text{II}}$ we have genus 1 from the Hamiltonian system~\eqref{eq:syst_93} and the area equal to 2. {We stress that despite being non-Hamiltonian with respect to the standard $2$-form, some of the systems under study in this work are Hamiltonian with respect to the pull back of the $2$-form induced by some changes of variables, i.e.\ system IX.B(2) and system XIII. In these cases, it is possible to introduce a suitable polynomial Hamiltonian function and then refer to the usual notion of Newton polygons, as in~\eqref{eq:Ham_Bu92} and~\eqref{eq:Ham_Bu13} respectively. }

In Section~\ref{sec:P2} with the method of the iterative polynomial regularisation, we encountered Hamiltonians of genus 2. We might ask whether these genus 2 Hamiltonians indicate a sort of ``Riccati extension'' or some other differential extension of the Painlev\'e equation in the specific case of quadratic systems. However, some genus 2 Hamiltonians also appear in relation to quartic Hamiltonians for the quasi-Painlev\'e case, as studied in~\cite{MDTK2}, therefore the meaning of genus 2 Hamiltonian needs to be clarified. 

Another important observation made in Section~\ref{sec:P2} concerns the possibility of obtaining a Hamiltonian system via the iterative regularisation starting from a system that is non-Hamiltonian. We can ask whether starting from a non-Hamiltonian system with some Painlev\'e transcendents hidden in the coefficients is it possible to generate a Hamiltonian system for the same Painlev\'e transcendent, and if so whether there is a way to predict this phenomenon. 

Other relevant aspects concern the integrability and the discrete version of these systems. It is not currently known whether the Bureau-Guillot systems here studied and possessing the Painlev\'e property, can be described by a suitable Lax pair, or if they can be put in relation to the systems studied in~\cite{Pickering}. It is also not known whether it is possible to construct a discrete analogue of these systems, such that the Painlev\'e property is hidden in some coefficients or appears as a regularisation condition. Being the systems under scrutiny quadratic, one of the natural choices would be to use the Kahan-Hirota-Kimura discretization \cite{Kahan1993,KahanLi1997,HirotaKimura2001,KimuraHirota2001}. 
This method is known for preserving the integrability of underlying continuous systems in many cases \cite{PetreraPfadlerSuris2009,CelledoniMcLachlanOwrenQuispel2013,CelledoniMcLachlanOwrenQuispel2014}.
Specifically, we note that a comprehensive geometric explanation of this phenomenon in the planar case has been provided by~\cite{PSS2019,GMcLQ_dummies}.



Finally, as already mentioned in the introduction, the Bureau-Guillot systems constitute some very interesting examples for the development of value distribution theory. The coefficient functions of most of these systems are meromorphic functions and are related to the Painlev\'e transcendents. The study of such systems is challenging from the point of view of the Nevanlinna theory. This is examined in  \cite{ECGF} in detail.

\section*{Acknowledgements}
The authors would like to thank G. Gubbiotti {and A. Stokes} for interesting insights and discussion. {The authors are also grateful to A. Guillot for sharing the revised version of the manuscript \cite{Guillot}.} The work of GF is part of the project ``ERDF A way of making Europe'', and the project PID2021-124472NB-I00 funded by MICIU/AEI/  10.13039/501100011033.

\section*{Statement \& Declaration}
The work of GF is part of the project ``ERDF A way of making Europe'', and the project PID2021-124472NB-I00 funded by MICIU/AEI/  10.13039/501100011033. 
The authors have no relevant financial or non-financial interests to disclose.
All the authors contributed equally to the study conception and design of the paper.  

No new data were created or analysed in this study. Data sharing is not applicable to this article.

\end{document}